\newif\ifstoc
\stoctrue 
\stocfalse

\ifstoc
  \documentclass{sig-alternate}
  \usepackage[obey]{fixacm}
  \usepackage{times}
\else
  \documentclass[letterpaper,9pt]{article}

  \usepackage{typearea}
  \typearea{15}

  \usepackage[compact]{titlesec}

\fi

\newcommand{\stocoption}[2]{{\ifstoc #1 \else #2 \fi}}

\newenvironment{packed_item}{
\begin{itemize}
  \setlength{\itemsep}{0.5pt}
  \setlength{\parskip}{0pt}
  \setlength{\parsep}{0pt}
}{\end{itemize}}

\newenvironment{packed_enum}{
\begin{enumerate}
  \setlength{\itemsep}{0.5pt}
  \setlength{\parskip}{0pt}
  \setlength{\parsep}{0pt}
}{\end{enumerate}}

\usepackage{theorem,latexsym,graphicx}
\usepackage{appendix}
\usepackage{amsfonts,amsmath,amssymb,dsfont}
\usepackage{cite,url,balance}
\usepackage{algorithm,algorithmic}
\usepackage{graphicx,psfrag}
\usepackage[usenames,dvipsnames]{color}
\usepackage{array,multirow}

\ifstoc
\else
\fi

\newtheorem{question}{Question}

\ifstoc
\else
\newenvironment{proof}{

\noindent{\bf Proof:}}
{\hfill$\blacksquare$

}
\fi

\newenvironment{proofof}[1]{

\noindent{\bf Proof of {#1}:}}
{\hfill$\blacksquare$

}

\newcounter{note}[section]
\renewcommand{\thenote}{\thesection.\arabic{note}}

\newcommand{\ktnote}[1]{\refstepcounter{note}$\ll${\bf Kunal~\thenote:} {\sf #1}$\gg$\marginpar{\tiny\bf KT~\thenote}}

\newtheorem{theorem}{Theorem}
\newtheorem{definition}{Definition}

\newtheorem{cor}{Corollary}

\newtheorem{corollary}[theorem]{Corollary}
\newtheorem{lemma}{Lemma}

\def\rank{{\rm rank }} 
\def\tr{{\rm tr}} 
\def\E{\mathbb{E}} 
\def\Pr{{\rm Pr}} 
\def\Var{{\rm Var}} 

\def\R{{\mathds{R}}} 

\newcommand{\eps}{\varepsilon}
\renewcommand{\vec}[1]{\mathbf{#1}}

\newcommand{\junk}[1]{}

\def\b0{{\bf 0}}

\newcommand{\mech}{\mathcal{M}}

\DeclareMathOperator{\inv}{inv}
\DeclareMathOperator{\vol}{vol}
\DeclareMathOperator{\vrad}{vrad}
\DeclareMathOperator{\sym}{sym}
\DeclareMathOperator{\disc}{disc}
\DeclareMathOperator{\herdisc}{herdisc}
\DeclareMathOperator{\vecdisc}{vecdisc}
\DeclareMathOperator{\hervecdisc}{hervecdisc}
\DeclareMathOperator{\lindisc}{lindisc}
\DeclareMathOperator{\vollb}{volLB}
\DeclareMathOperator{\detlb}{detLB}
\DeclareMathOperator{\specLB}{specLB}
\DeclareMathOperator{\opt}{opt}
\DeclareMathOperator{\err}{err}
\DeclareMathOperator{\diag}{diag}
\DeclareMathOperator{\supp}{supp}
\DeclareMathOperator{\vspan}{span}
\def\polylog{\operatorname{polylog}}
\newcommand{\cut}[1]{}

\newcommand{\epsdp}{\ensuremath{\eps}-DP}
\newcommand{\epsdeldp}{\ensuremath{(\eps,\delta)}-DP}

\begin{document}
\title{The Geometry of Differential Privacy:\\
The Sparse and Approximate Cases}
\author{Aleksandar Nikolov\thanks{Department of Computer Science,
    Rutgers University, Piscataway, NJ 08854. This work was done while
  the author was at Microsoft Research SVC.} \and Kunal Talwar\thanks{Microsoft Research SVC, Mountain View, CA 94043.} \and Li Zhang\thanks{Microsoft Research SVC, Mountain View, CA 94043.}}
\date{}
\maketitle

\begin{abstract}

In this work, we study trade-offs between accuracy and privacy in the context of linear queries over histograms. This is a  rich class of queries that includes contingency tables and range queries, and has been a focus of a long line of work~\cite{BlumLR08,RothR10,DworkRV10,HardtT10,HardtR10,LiHRMM10,12vollb}. For a given set of $d$ linear queries over a database $x \in \R^N$, we seek to find the differentially private mechanism that has the minimum mean squared error. For pure differential privacy, \cite{HardtT10,12vollb} give an $O(\log^2 d)$ approximation to the optimal mechanism. Our first contribution is to give an $O(\log^2 d)$ approximation guarantee for the case of $(\eps,\delta)$-differential privacy. Our mechanism is simple, efficient and adds carefully chosen correlated Gaussian noise to the answers. We prove its approximation guarantee relative to the {\em hereditary discrepancy} lower bound of~\cite{MNstoc}, using tools from convex geometry.

We next consider this question in the case when the number of queries exceeds the number of individuals in the database, i.e. when  $d > n \triangleq \|x\|_1$. The lower bounds used in the previous approximation algorithm no longer apply, and in fact better mechanisms are known in this setting~\cite{BlumLR08,RothR10,HardtR10,GuptaHRU11,GuptaRU11}. Our second main contribution is to give an $(\eps,\delta)$-differentially private mechanism that for a given query set $A$ and an upper bound $n$ on $\|x\|_1$, has mean squared error within $\polylog(d,N)$ of the optimal for $A$ and $n$. This approximation is achieved by coupling the Gaussian noise addition approach with linear regression over the $\ell_1$ ball. Additionally, we show a similar polylogarithmic approximation guarantee for the best $\eps$-differentially private mechanism in this sparse setting. Our work also shows that for arbitrary counting queries, i.e. $A$ with entries in $\{0,1\}$, there is an $\eps$-differentially private mechanism with expected error $\tilde{O}(\sqrt{n})$ per query, improving on the $\tilde{O}(n^{\frac{2}{3}})$ bound of~\cite{BlumLR08}, and matching the lower bound implied by~\cite{DinurN03} up to logarithmic factors.

The connection between  hereditary discrepancy and the privacy mechanism enables us to derive the first polylogarithmic approximation to the hereditary discrepancy of a matrix $A$.

\junk{Interestingly, computing hereditary discrepancy exactly is not known to be in NP.}
\end{abstract}

\section{Introduction}
Differential privacy~\cite{DMNS} is a recent privacy definition that has quickly become the standard notion of privacy in statistical databases. Informally, a mechanism (a randomized function on databases) satisfies differential privacy if the distribution of the outcome of the mechanism does not change noticeably when one individual's input to the database is changed. Privacy is measured by how small this change must be: an $\eps$-differentially private (\epsdp) mechanism $\mech$ satisfies $\Pr[\mech(x) \in S] \leq \exp(\eps) \Pr[\mech(x\rq{})\in S]$ for any pair $x,x\rq{}$ of neighboring databases, and for any measurable subset $S$ of the range. A relaxation of this definition is {\em approximate differential privacy}. A mechanism $\mech$ is $(\eps,\delta)$-differentially private (\epsdeldp) if $\Pr[\mech(x) \in S] \leq \exp(\eps) \Pr[\mech(x\rq{})\in S] + \delta$ with $x,x\rq{},S$ as before. Here $\delta$ is thought of as negligible in the size of the database. Both these definitions satisfy several desirable properties such as composability, and are resistant to post-processing of the output of the mechanism. 

In recent years, a large body of research has shown that this strong privacy definition still allows for very accurate analyses of statistical databases. At the same time, answering a large number of adversarially chosen queries accurately is inherently impossible with any semblance of privacy. Indeed Dinur and Nissim~\cite{DinurN03} show that answering $\tilde{O}(d)$ random subset sums ($\tilde{O}$ hides polylogarithmic factors in $N,d,1/\delta$.) of a set of $d$ bits with (per query) error $o(\sqrt{d})$ allows an attacker to reconstruct (an arbitrarily good approximation to) all the private information. Thus there is an inherent trade-off between privacy and accuracy when answering a large number of queries. In this work, we study this trade-off in the context of counting queries, and more generally linear queries.

We think of the database as being given by a multiset of database
rows, one for each individual. We will let $N$ denote the size of the
universe that these rows come from, and we will denote by $n$ the
number of individuals in the database. We can represent the database
as its histogram $x \in \R^N$ with $x_i$ denoting the number of
occurrences of the $i$th element of the universe. Thus $x$ would in
fact be a vector of non-negative integers with $\|x\|_1 = n$. We will
be concerned with reporting reasonably accurate answers to a given set
of $d$ {\em linear} queries over this histogram $x$. This set of
queries can naturally be represented by a matrix $A \in \R^{d\times
  N}$ with the vector $Ax \in \R^d$ giving the correct answers to the
queries.  When $A \in \{0,1\}^{d\times N}$, we call such  queries
counting queries. We are interested in the (practical) regime where
$N\gg d\gg n$, although our results hold for all settings of the
parameters. 

A differentially private mechanism will return a noisy answer to the query $A$ and, in this work, we measure the performance of the mechanisms in terms of its worst case total expected squared error. Suppose that $X\subseteq \R^{N}$ is the set of all possible databases. The error of a mechanism $\mech$ is defined as
 $\err_\mech(A,X) = \max_{x\in X} \E[\|\mech(x)-Ax\|_2^2]$. Here the expectation is taken over the internal coin tosses of the mechanism itself, and we look at the worst case of this expected squared error over all the databases in $X$.  Unless stated otherwise, the error of the mechanism will refer to this worst case expected $\ell_2^2$ error. Phrased thus, the Gaussian noise mechanism of Dwork et al.~\cite{DworkKMMN06} gives error at most $O(d^{2})$ for any counting query and guarantees \epsdeldp\footnote{Here and in the rest of the introduction, we suppress the dependence of the error on $\eps$ and $\delta$.} over all the databases, i.e. $X=\R^N$.  Moreover, the aforementioned lower bounds imply that there exist counting queries for which this bound can not be improved. For \epsdp, Hardt and Talwar~\cite{HardtT10} gave a mechanism with error $O(d^{2}\log \frac{N}{d})$ and showed that this is the best possible for random counting queries. Thus the worst case accuracy for counting queries is fairly well-understood in this measure.

 Specific sets of counting queries of interest can however admit much
 better mechanisms than adversarially chosen queries for which the
 lower bounds are shown. Indeed several classes of specific queries
 have attracted attention. Some, such as range queries, are
 \lq\lq{}easier\rq\rq{}, and asymptotically better mechanisms can be
 designed for them. Others, such as constant dimensional contingency
 tables, are nearly as hard as general counting queries, and
 asymptotically better mechanisms can be ruled out in some ranges of
 the parameters.  These query-specific upper bounds are usually proved
 by carefully exploiting the structure of the query, and
 query-specific lower bounds have been proved by reconstruction
 attacks that exploit a lower bound on the smallest singular value of
 an appropriately chosen
 $A$~\cite{DinurN03,DworkMT07,DworkY08,KasiviswanathanRSU10,De12,KasiviswanathanRS13}. It
 is natural to address this question in a competitive analysis
 framework: can we design an efficient algorithm that given any query
 $A$, computes (even approximately) the minimum error differentially
 private mechanism for $A$?

Hardt and Talwar~\cite{HardtT10} answered this question in the affirmative for \epsdp\ mechanisms, and gave a mechanism that has error within factor $O(\log^3 d)$ of the optimal assuming a conjecture from convex geometry known as the hyperplane conjecture or the slicing conjecture. Bhaskara et al.~\cite{12vollb} removed the dependence on the hyperplane conjecture and improved the approximation ratio to $O(\log^2 d)$. Can relaxing the privacy requirement to \epsdeldp\ help with accuracy? In many settings, \epsdeldp\ mechanisms can be simpler and more accurate than the best known \epsdp\ mechanisms.  This motivates the first question we address.
\begin{question}
\label{ques:optdense}
Given $A$, can we efficiently approximate the optimal error \epsdeldp\ mechanism for it? 
\end{question}
Hardt and Talwar~\cite{HardtT10} showed that for some $A$, the lower bound for \epsdp\ mechanism can be $\Omega(\log N)$ larger than known \epsdeldp\ mechanisms. For non-linear Lipschitz queries, De~\cite{De12} showed that this gap can be as large as $\Omega(\sqrt{d})$ (even when $N=d$). This leads us to ask:
\begin{question}
How large can the gap between the optimal \epsdp\ mechanism and the optimal \epsdeldp\ mechanism be for linear queries?
\end{question}

When the databases are sparse, e.g. when $\|x\|_1\leq n \ll d$, one may obtain better mechanisms.  Blum, Ligett and Roth~\cite{BlumLR08} gave an \epsdp\ mechanism that can answer any set of $d$ counting queries  with error $\tilde{O}(dn^{\frac  4 3})$. A series of subsequent works~\cite{DworkNRRV09,DworkRV10,RothR10,HardtR10,GuptaHRU11,HardtLM12,GuptaRU11} led to \epsdeldp\ mechanisms that have error only $\tilde{O}(dn)$. Thus when $n < d$, the lower bound of $O(d^2)$ for arbitrary databases can be breached by exploiting the sparsity of the database. This motivates a more refined measure of error that takes the sparsity of $A$ into account. Given an $A$ and $n$, one can ask for the mechanism $\mech$ that minimizes the sparse case error $\max_{x: \|x\|_1 \leq n} \E[\|\mech(x) - Ax\|_2^2]$. The next set of questions we study address this measure.
\begin{question}
Given $A$ and $n$, can we approximate the optimal sparse case error \epsdeldp\ mechanism for $A$ when restricted to databases of size at most $n$?
\end{question}

\begin{question}
Given $A$ and $n$, can we approximate the optimal sparse case error \epsdp\ mechanism for $A$ when restricted to databases of size at most $n$?
\end{question}

The gap between the $\tilde{O}(dn^{\frac  4 3})$ error \epsdp\ mechanism of~\cite{BlumLR08} and the $\tilde{O}(dn)$ error \epsdeldp\ mechanism of~\cite{HardtR10} leads us to ask:
\begin{question}
  Is there an \epsdp\ mechanism with error $\tilde{O}(dn)$ for
  databases of size at most $n$?
\label{ques:isblroptimal}
\end{question}

\subsection{Results}
In this work, we answer Questions~\ref{ques:optdense}-\ref{ques:isblroptimal} above. Denote by $B_1^N\triangleq \{ x\in \R^N\;:\; \|x\|_1\leq 1\}$ the $N$-dimensional $\ell_1$ ball. Recall that for any query matrix $A\in \R^{d\times N}$ and any set $X\subseteq \R^N$, the (worst-case expected squared) error of $\mech$ is defined as  
\[\err_\mech(A,X) \triangleq \max_{x\in X} \E[\|\mech(x)-Ax\|_2^2]\,.\]

In this paper, we are interested in both the case when $X=\R^N$, called the \emph{dense} case, and when $X=nB_1^N$ for $n<d$, called the \emph{sparse} case. We also write $\err_\mech(A) \triangleq \err_\mech(A,\R^N)$ and $\err_\mech(A,n) \triangleq \err_\mech(A,nB_1^N)$. 

Our first result is a simple and efficient mechanism that for  query matrix $A$ gives an $O(\log^2 d)$ approximation to the optimal error.
\begin{theorem}
Given a query matrix $A \in \R^{d\times N}$, there is an efficient \epsdeldp\ mechanism $\mech$ and an efficiently computable lower bound $L_A$ such that
\begin{packed_item}
\item $\err_M(A) \leq O(\log^2 d\log 1/\delta) \cdot L_A$, and
\item for any \epsdeldp\ mechanism $\mech\rq{}$, $\err_{\mech\rq{}}(A,d)\geq  L_A$.
\end{packed_item}
\label{thm:epsdeldense}
\end{theorem}

We also show that the gap of $\Omega(\log (N/d))$ between \epsdp\ and \epsdeldp\ mechanisms shown in~\cite{HardtT10} is essentially the worst possible, within $\polylog(d)$ factor, for linear queries. More precisely, the lower bound on \epsdp\ mechanisms used in~\cite{HardtT10} is always within $O(\log (N/d) \polylog(d))$ of the lower bound $L_A$ computed by our algorithm above.  Let $\mech^*$  denote the \epsdp\ generalized $K$-norm mechanism in~\cite{HardtT10}.
\begin{theorem}\label{thm:epsdelgap}
  For any \epsdeldp\ mechanism $\mech$, $\err_{\mech}(A)=\Omega(1/(\log^{O(1)}(d)\log (N/d)))\err_{\mech^*}(A)$.
\end{theorem}

We next move to the sparse case. Here we give results analogous to the dense case with a slightly worse approximation ratio.
\begin{theorem}
Given $A \in \R^{d\times N}$ and a bound $n$, there is an efficient \epsdeldp\ mechanism $\mech$ and an efficiently computable lower bound $L_{A,n}$ such that
\begin{packed_item}
\item $\err_{\mech}(A, n) \leq O(\log^{3/2} d \cdot \sqrt{\log N\log 1/\delta}+\log^2 d \log 1/\delta) \cdot L_{A,n}$, and
\item For any \epsdeldp\ mechanism $\mech\rq{}$, $\err_{\mech\rq{}}(A, n)\geq  L_{A,n}$.
\end{packed_item}
\label{thm:epsdelsparse}
\end{theorem}

\begin{theorem}
Given $A \in \R^{d\times N}$ and a bound $n$, there is an efficient \epsdp\ mechanism $\mech$ and an efficiently computable lower bound $L_{A,n}$ such that
\begin{packed_item}
\item $\err_{\mech}(A, n) \leq O(\log^{O(1)} d \cdot\log^{3/2} N) \cdot L_{A,n}$, and
\item For any \epsdp\ mechanism $\mech\rq{}$, $\err_{\mech\rq{}}(A, n)\geq  L_{A,n}$.
\end{packed_item}
\label{thm:epssparse}
\end{theorem}

We remark that in these theorems, our upper bounds hold for all $x$ with $\|x\|_1 \leq n$, whereas the lower bounds hold even when $x$ is an integer vector.

The \epsdeldp\ mechanism of Theorem~\ref{thm:epsdelsparse} when run on any counting query has error no larger than the best known bounds~\cite{GuptaRU11} for counting queries, up to constants (not ignoring logarithmic factors). The \epsdp\ mechanism of Theorem~\ref{thm:epssparse} when run on any counting query can be shown to have nearly the same asymptotics, answering question~\ref{ques:isblroptimal} in the affirmative.
\begin{theorem}\label{thm:countingeps}
For any counting query $A$, there is an \epsdp\ mechanism $\mech$ such that $\err_{\mech}(A,n) = \tilde{O}(dn)$.
\end{theorem}

We will summarize some key ideas we use to achieve these results. More details will follow in Section~\ref{sec:technique}.

For the upper bounds, the first crucial step is to decompose $A$ into ``geometrically  nice'' components and then add Gaussian noise to each component. This is similar to the approach in~\cite{HardtT10,12vollb} but we use the minimum volume enclosing ellipsoid, rather than the $M$-ellipsoid used in those works, to facilitate the decomposition process. This allows us to handle the approximate and the sparse cases. In addition, it simplifies the mechanism as well as the analysis. For the sparse case, we further couple the mechanism with  least squares estimation of the noisy answer with respect to $nAB_1^N$.  By utilizing techniques from statistical estimation, we can show that this process can reduce the error when $n <d$, and prove an error upper bound dependent on the size of the smallest projection of $nAB_1^N$.

For the lower bounds, we first lower bound the accuracy of $(\epsilon,\delta)$-DP mechanism by the hereditary discrepancy of the query matrix $A$, which we in turn lower bound in terms of the least singular values of submatrices of $A$. Finally, we close the loop by utilizing the restricted invertibility principle by Bourgain and Tzafriri~\cite{bour-tza} and its extension by Vershynin~\cite{vershynin} which, informally, shows that if there does not exist a ``small'' projection of $nAB_1^N$ then $A$ has a ``large'' submatrix with a ``large'' least singular value.
\cut{
The algorithms of Theorems~\ref{thm:epsdeldense}
and~\ref{thm:epsdelsparse} are more efficient and conceptually simpler
than their \epsdp\ counterparts. The main computational building blocks
for the algorithm of Theorem~\ref{thm:epsdeldense} are computing a
(approximate) minimum volume enclosing ellipsoid of a convex polytope
and sampling from a multidimensional correlated Gaussian
distribution. Both problems are classically studied and have
algorithms with running time linear in $N$ and polynomial in $d$. The
algorithm of Theorem~\ref{thm:epsdelsparse} additionally needs to
compute the point inside a convex polytope closest to a given
point. This problem is also classical in statistics; it can be approximated by $O(n)$ rounds of a simple gradient descent algorithm to within error no higher than the error already incurred due to privacy. Both algorithms achieve privacy by adding noise from a distribution oblivious to the private database, and, therefore, interaction with private data can be limited to a small and simple part of the algorithm. }

\medskip\noindent{\bf Approximating Hereditary Discrepancy}

The discrepancy of a matrix $A \in \R^{d\times N}$ is defined to be $\disc(A) = \min_{x \in \{-1,+1\}^N} \|Ax\|_{\infty}$. The hereditary discrepancy of a matrix is defined as $\herdisc(A) = \max_{S \subseteq [N]} \disc(A|_{S})$, where $A|_{S}$ denotes the matrix $A$ restricted to the columns indexed by $S$.

As hereditary discrepancy is a maximum over exponentially many
submatrices, it is not a priori clear if there even exists a
polynomial-time verifiable certificate for low hereditary
discrepancy. Additionally, we can show that it is $\mathsf{NP}$-hard
to approximate hereditary discrepancy to within a factor of
$3/2$. Bansal~\cite{Bansal10} gave a pseudo-approximation algorithm
for hereditary discrepancy, which efficiently computes a coloring of
discrepancy at most a factor of $O(\log dN)$ larger than $\herdisc(A)$
for a $d\times N$ matrix $A$. His algorithm allows efficiently
computing a lower bound on $\herdisc$ for any restriction $A|_S$;
however, such a lower bound may be arbitrarily loose, and before our
work it was not known how to efficiently compute nearly matching lower
and upper bounds on
$\herdisc$.  

Muthukrishnan and Nikolov~\cite{MNstoc} show that for a query matrix $A \in \R^{d\times N}$, the error of any \epsdeldp\ mechanism is lower bounded by (an $\ell_2^2$ version of) $(\herdisc(A))^2$ (up to logarithmic factors). Moreover, the lower bound used in Theorem~\ref{thm:epsdeldense} is in fact a lower bound on this version of $\herdisc(A)$. Using the von Neumann minimax theorem, we can go between the $\ell_2^2$ and the $\ell_\infty$ versions of these concepts, allowing us to sandwich the hereditary discrepancy of $A$ between two quantities: a determinant based lower bound and the efficiently computable  expected error of the private mechanism. As the two quantities are nearly matching, our work therefore leads to a polylogarithmic approximation to the hereditary discrepancy of any matrix $A$.

\subsection{Techniques}\label{sec:technique}
In addition to known techniques from the differential privacy literature, our work borrows tools from discrepancy theory, convex geometry and statistical estimation. We next briefly describe how they fit in.

Central to designing a provably good approximation algorithm is an efficiently computable lower bound on the optimum. Muthukrishnan and Nikolov~\cite{MNstoc} proved that (a slight variant of) the hereditary discrepancy of $A$ leads to a lower bound for the error of any \epsdeldp\ mechanism. Lov\'{a}sz, Spencer and Vesztergombi~\cite{lovasz1986discrepancy} showed that  hereditary discrepancy itself can be lower bounded by a quantity called the determinant lower bound. Geometrically, this lower bound corresponds to picking the $d$ columns of $A$ that (along with the origin) give us a simplex with the largest possible volume. The volume or this simplex, appropriately normalized, gives us a lower bound on OPT. More precisely for any simplex $S$, $d^3 \cdot \vol(S)^{\frac{2}{d}} \log^2 d$ gives a lower bound on the error. The $\log^2 d$ factor can be removed by using a lower bound based on the least singular values of submatrices of $A$. Geometrically, for the least singular value lower bound we need to find a simplex of large volume whose $d$ non-zero vertices are also nearly pairwise orthogonal. 

If the $N$ columns of $A$ all lie in a unit ball of radius $R$, it can be shown that adding Gaussian noise proportional to $R$ suffices to guarantee \epsdeldp, resulting in a mechanism having total squared error $dR^2$. Can we relate this quantity to the lower bound? It turns out that if the unit ball of radius $R$ is the minimum volume ellipsoid containing the columns of $A$, this can be done. In this case, a result of Vershynin~\cite{vershynin}, building on the restricted invertability results by Bourgain and Tzafriri~\cite{bour-tza}, tells us that one can find $\Omega(d)$ vertices of $K$ that touch the minimum containing ellipsoid, and are nearly orthogonal. The simplex formed by these vertices therefore has large volume, giving us a \epsdeldp\ lower bound of $\Omega(dR^2)$. In this case, the Gaussian mechanism with the optimal $R$ is within a constant factor of the lower bound. When the minimum volume enclosing ellipsoid is not a ball, we need to project the query along the $\frac{d}{2}$ shortest axes of this ellipsoid, answer this projection using the Gaussian mechanism, and recurse on the orthogonal projection. Using the full power of the restricted invertability result by Vershynin allows us to construct a large simplex and prove our competitive ratio.

Hardt and Talwar~\cite{HardtT10} also used a volume based lower bound, but for \epsdp\ mechanisms, one can take $K$, the symmetric convex hull of all the columns of $A$ and use its volume instead of the volume of $S$ in the lower bound above. How do these lower bounds compare? By a result of B\'ar\'any and F\"uredi~\cite{barany} and Gluskin~\cite{gluskin}, one can show that the volume of the convex hull of $N$ points can be bounded by $(\log N)^{d/2} d^{-d/2}$ times that of the minimum enclosing ellipsoid. This, along with the aforementioned restricted invertability results, allows us to prove that the \epsdp\ lower bound is within $O((\log N) \polylog d)$ of the \epsdeldp\ lower bound.

How do we handle sparse queries? The first observation is that the lower bounding technique gives us $d$ columns of $A$  and the resulting lower bound holds not just for $A$ but even for the  $d\times d$ submatrix of $A$ corresponding to the maximum volume simplex $S$; moreover, the lower bound holds even when all databases are restricted to $O(d)$ individuals. Thus the lower bound holds when $n=O(d)$ and this value marks the transition between the sparse and the dense cases. Moreover, when the minimum volume ellipsoid containing the columns of $A$ is a ball,  the restricted invertibility principle of Bourgain and Tzafriri and Vershynin gives us a $d$-dimensional simplex with nearly pairwise orthogonal vertices, and, therefore any $n$-dimensional face of this simplex is another simplex of large volume. The large $n$-dimensional simplex gives a lower bound on error when databases are restricted to have at most $n$ individuals. 

For smaller $n$, the error added by the Gaussian mechanism may be too large, and even though the value $Ax$ lies in $nAB_1^N$, the noisy  answer will likely fall outside this set. A common technique in statistical estimation for handling such error is to \lq\lq{}project\rq\rq{} the noisy point back into $nAB_1^N$, i.e. report the point $\hat{y}$ in $nAB_1^N$ that minimizes the Euclidean distance to the noisy answer $\tilde{y}$. This projection step provably reduces the expected error! Geometrically, we use well known techniques from statistics to show that the error after projection is bounded by the ``shadow'' that $nAB_1^N$ leaves on the noise vector; this shadow is much smaller than the length of the noise vector when $n = o(d)$. In fact, when the noise is a spherical Gaussian, it can be shown that $\|\hat{y}-y\|_2^2$ is only about $\frac{n}{d}\|\tilde{y}-y\|_2^2$.  
This gives near optimal bounds for the case when the minimum volume ellipsoid is a ball; the general case is handled using a recursive mechanism as before.

To get an \epsdp\ mechanism, we use the $K$-norm mechanism~\cite{HardtT10} instead of Gaussian noise. To bound the shadow of $nAB_1^N$ on $w$, where $w$ is the noise vector generated by the $K$-norm mechanism, we first analyze the expectation of $\langle a_i, w\rangle$ for any column of $A$, and we use the log concavity of the noise distribution to prove concentration of this random variable. A union bound helps complete the argument as in the Gaussian case.

\subsection{Related Work}

Dwork et al.~\cite{DMNS} showed that any query can be released while adding noise proportional to the total {\em sensitivity} of the query. This motivated the question of designing mechanisms with good guarantees for any set of low sensitivity queries. Nissim, Raskhodnikova and Smith~\cite{NissimRS07} showed that adding noise proportional to (a smoothed version of) the {\em local sensitivity} of the query suffices for guaranteeing differential privacy; this may be much smaller than the worst case sensitivity for non-linear queries. Lower bounds on the amount of noise needed for general low sensitivity queries have been shown in~\cite{DinurN03, DworkMT07, DworkY08, DMNS,RastogiSH07, HardtT10,De12}.  Kasiviswathan et al.~\cite{KasiviswanathanRSU10} showed upper and lower bounds for contingency table queries and more recently~\cite{KasiviswanathanRS13} showed lower bounds on publishing error rates of classifiers or even M-estimators. Muthukrishnan and Nikolov~\cite{MNstoc} showed that combinatorial discrepancy lower bounds the noise for answering any set of linear queries. 

Using learning theoretic techniques, Blum, Ligett and Roth~\cite{BlumLR08} first showed that one can exploit sparsity of the database, and answer a large number of counting queries with error small compared to the number of individuals in the database. This line of work has been further extended and improved in terms of error bounds, efficiency, generality and interactivity in several subsequent works~\cite{DworkNRRV09,DworkRV10, RothR10, HardtR10,GuptaHRU11,HardtLM12}.

Ghosh, Roughgarden and Sundarajan~\cite{GhoshRS09} showed that for any one dimensional counting query, a discrete version of the Laplacian mechanism is optimal for pure privacy in a very general utilitarian framework and Gupte and Sundararajan~\cite{GupteS10} extended this to risk averse agents. Brenner and Nissim~\cite{BrennerN10} showed that such universally optimal private mechanisms do not exist for two counting queries or for a single non-binary sum query.  As mentioned above, Hardt and Talwar~\cite{HardtT10}, and Bhaskara et al.~\cite{12vollb} gave relative guarantees for multi-dimensional queries under pure privacy with respect to total squared error. De~\cite{De12} unified and strengthened these bounds and showed stronger lower bounds for the class of non-linear low sensitivity queries.

For specific queries of interest, improved upper bounds are known. Barak et al.~\cite{BarakCDMK07} studied low dimensional marginals and showed that by running the Laplace mechanism on a different set of queries, one can reduce error. Using a similar strategy, improved mechanisms were given by~\cite{XiaoWG10,ChanSS10} for orthogonal counting queries, and near optimal mechanisms were given by  Muthukrishnan and Nikolov~\cite{MNstoc} for  halfspace counting queries. The  approach of answering a set of queries different from the target query set has also been studied in more generality and for other sets of queries by~\cite{LiHRMM10,DingWHL11,RastogiSH07,XiaoWG10,XiaoXY10,YuanZWXYH12}. Li and Miklau~\cite{LiM12a,LiM12b} study a class of mechanisms called extended matrix mechanisms and show that one can efficiently find the best mechanisms from this class. Hay et al.~\cite{HayRMS10} show that in certain settings such as unattributed histograms, correcting noisy answers to enforce a consistency constraint can improve accuracy. 

Very recently, Fawaz et al.~\cite{convolutions} used the hereditary discrepancy lower bounds of Muthukrishnan and Nikolov, as well as the determinant lower bound on discrepancy of Lovasz, Spencer, and Vesztergombi, to prove that a certain Gaussian noise mechanism is nearly optimal (in the dense setting) for computing any given convolution map. Like our algorithms, their algorithm adds correlated Gaussian noise; however, they always use the Fourier basis to correlate the noise.

We refer the reader to texts by Chazelle~\cite{Chazelle} and Matou\v{s}ek~\cite{Matousek} and the chapter by Beck and S\'{o}s~\cite{beck-sos} for an introduction to discrepancy theory. Bansal~\cite{Bansal10} showed that a semidefinite relaxation can be used to design a pseudo-approximation algorithm for hereditary discrepancy. Matou\v{s}ek~\cite{Matousek11} showed that the determinant based lower bound of Lov\'{a}sz, Spencer and Vesztergombi~\cite{lovasz1986discrepancy} is tight up to polylogarithmic factors. Larsen~\cite{disc-larsen} showed applications of hereditary discrepancy to data structure lower bounds, and Chandrasekaran and  Vempala~\cite{disc-ip} recently showed applications of hereditary discrepancy to problems in integer programming.


%
%
%
%
%
%
%
%
%
%
%
%

\paragraph{Roadmap.} In Section~\ref{sect:prelims} we introduce
relevant preliminaries. In Section~\ref{sect:apx} we present our main
results for approximate differential privacy, and in
Section~\ref{sect:pure} we present our main results for pure
differential privacy. In Section~\ref{sect:universal} we prove
absolute upper bounds on the error required for privately answering
sets of $d$ counting queries. In Section~\ref{sect:ext} we give some
extensions and applications of our main results, namely an optimal efficient
mechanism for $\ell_\infty$ error in the dense case, and the
efficient approximation to hereditary discrepancy implied by that
mechanism. We conclude in Section~\ref{sect:concl}. 

\section{Preliminaries}
We start by introducing some basic notation. 

Let $B_1^d$, and $B_2^d$ be, respectively, the $\ell_1$ and $\ell_2$
unit balls in $\R^d$.  Also, let $\sym\{a_1, \ldots a_N\}$ be the
convex hull of the vectors $\pm a_1, \ldots, \pm a_N$. Equivalently,
$\sym\{a_1, \ldots, a_N\} = AB_1^N$ where $A$ is a matrix whose columns
equal $a_1, \ldots, a_N$.

For a $d \times N$ matrix $A$ and a set $S \subseteq [N]$, we denote
by $A|_S$ the submatrix of $A$ consisting of those columns of $A$
indexed by elements of $S$. Occasionally we refer to a matrix $V$
whose columns form an orthonormal basis for some subspace of interest
$\mathcal{V}$ as the orthonormal basis of
$\mathcal{V}$. $\mathcal{P}_k$ is the set of orthogonal projections
onto $k$-dimensional subspaces of $\R^d$.

By $\sigma_{\min}(A)$ and $\sigma_{\max}(A)$ we denote, respectively,
the smallest and largest singular value of $A$.\stocoption{}{ I.e.,
$\sigma_{\min}(A) = \min_{x: \|x\|_2 = 1}{\|Ax\|_2}$ and
$\sigma_{\max}(A) = \max_{x: \|x\|_2 = 1}{\|Ax\|_2}$.} In general,
$\sigma_i(A)$ is the $i$-th largest singular value of $A$, and
$\lambda_i(A)$ is the $i$-th largest eigenvalue of
$A$. \stocoption{}{We recall the minimax characterization of
  eigenvalues for symmetric matrices:
  \begin{equation*}
    \lambda_i = \max_{\mathcal{V}: \dim \mathcal{V} = i}\min_{x \in
      \mathcal{V}:\|x\|_2 = 1}{x^TAx}.
  \end{equation*}
For a matrix $A$ (and the corresponding linear operator), we denote by $\|A\|_2=\sigma_{\max}(A)$ the spectral norm of $A$ and $\|A\|_F = \sqrt{\sum_i \sigma_i^2(A)} = \sqrt{\sum_{i,j} a_{i,j}^2}$ the Frobenius norm of $A$. By $\ker A$ we denote the kernel of $A$, i.e. the subspace of vectors $x$ for which
$Ax = 0$. } 

\subsection{Geometry}


For a set $K \subseteq \R^d$, we denote by $\vol_d(K)$ its
$d$-dimensional volume. Often we use instead the \emph{volume radius}
\begin{equation*}
\vrad_d(K) \triangleq (\vol(K)/\vol(B_2^d))^{1/d}.
\end{equation*}
Subscripts are omitted when
this does not cause confusion.\stocoption{}{ When $K$ lies in a $k$-dimensional
affine subspace of $\R^d$, $\vol(K)$ and $\vrad(K)$ (without
subscripts) are understood to imply $\vol_k$ and $\vrad_k$,
respectively.}

For a convex body $K \subseteq \R^d$, the \emph{polar body} $K^\circ$ is
defined by $K^\circ = \{y: \langle y, x \rangle \leq 1~\forall x \in
K\}$. The fundamental fact about polar bodies we use is that for any
two convex bodies $K$ and $L$
\begin{equation}\label{eq:conv-duality}
  K \subseteq L \Leftrightarrow L^\circ \subseteq K^\circ.
\end{equation}
In the remainder of this paper, when we claim that a fact follows ``by
convex duality,'' we mean that it is implied by
(\ref{eq:conv-duality}). 

A convex body $K$ is \emph{(centrally) symmetric} if $-K = K$. The \emph{Minkowski
  norm} $\|x\|_K$ induced by a symmetric convex body $K$ is defined as 
$\|x\|_K \triangleq \min\{r \in \R: x \in rK\}$. The Minkowski norm induced by
the polar body $K^\circ$ of $K$ is the \emph{dual norm} of $\|x\|_K$
and also has the form 
  $\|y\|_{K^\circ} = \max_{x \in K}{\langle x,  y\rangle}$. 
For convex symmetric $K$, the induced norm and dual
norm satisfy H\"{o}lder's inequality:
\begin{equation}
  \label{eq:holder}
  |\langle x, y \rangle| \leq \|x\|_K \|y\|_{K^\circ}.
\end{equation}

An \emph{ellipsoid} in $\R^d$ is the image of $B_2^d$ under an affine
map. All ellipsoids we consider are symmetric, and therefore, are
equal to an image $F B_2^d$ of the ball $B_2^d$ under a linear map
$F$. A full dimensional ellipsoid $E = FB_2^d$ can be equivalently
defined as $E = \{x: x^T(FF^T)^{-1}x \leq 1\}$. \junk{ Moreover, without loss of generality we may assume that $F$ is
positive semidefinite.}  The polar body of a symmetric ellipsoid $E = F
B_2^d$ is the ellipsoid (or cylinder with an ellipsoid as its base in
case $F$ is not full dimensional) $E^\circ = \{x: x^TFF^Tx \leq
1\}$.


We repeatedly use a classical theorem of Fritz John, characterizing
the (unique) \emph{minimum volume enclosing ellipsoid} (MEE) of any convex body
$K$.\stocoption{}{ We note that John's theorem is frequently stated in terms of the
maximum volume enclosed ellipsoid in $K$; the two variants of the
theorem are equivalent by convex duality. The MEE of $K$ is also known
as a the L\"{o}wner or L\"{o}wner-John ellipsoid of $K$. }

\begin{theorem}[\cite{john}]\label{thm:john}
  Any convex body $K \subseteq \R^d$ is contained in a unique
  ellipsoid of minimal volume. This ellipsoid is $B_2^d$ if and only
  if there exist unit vectors $u_1, \ldots, u_m \in K \cap B_2^d$ and
  positive reals $c_1, \ldots, c_m$\junk{, $\sum{c_i} = d$,} such that
  \begin{align*}
    \sum{c_i u_i} &= 0\\
    \sum{c_i u_i u_i^T} &= I
  \end{align*}
\end{theorem}

According to John's characterization, when the MEE of $K$ is the ball
$B_2^d$, the contact points of $K$ and $B_2^d$ satisfy a structural
property --- the identity decomposes into a linear combination of the
projection matrices onto the lines of the contact points. Intuitively,
this means that $K$ ``hits'' $B_2^d$ in all directions --- it has to,
or otherwise $B_2^d$ can be ``pinched'' in order to produce a smaller
ellipsoid that still contains $K$.  This intuition is formalized by a
theorem of Vershynin, which generalizes the work
of Bourgain and Tzafriri on restricted
invertibility~\cite{bour-tza}.\stocoption{}{ Vershynin~(\cite{vershynin} Theorem~3.1) shows that
  there exist $\Omega(d)$ contact points of $K$ and $B_2^d$ which are
  approximately pairwise orthogonal.} 

\begin{theorem}[\cite{vershynin}]\label{thm:bt}
  Let $K \subseteq \R^d$ be a symmetric convex body whose minimum volume
  enclosing ellipsoid is the unit ball $B_2^d$. Let $T$ be a linear
  map with spectral norm $\|T\|_2 \leq 1$. Then for any $\beta$, there
  exist constant $C_1(\beta)$, $C_2(\beta)$ and contact points $x_1, \ldots, x_k$
  with $k \geq (1-\beta) \|T\|_F^2$ such that the matrix $TX =
  (Tx_i)_{i = 1}^k$ satisfies
  \begin{equation*}
    C_1(\beta) \frac{\|T\|_F}{\sqrt{d}} \leq \sigma_{\min}(TX) \leq
    \sigma_{\max}(TX) \leq    C_2(\beta) \frac{\|T\|_F}{\sqrt{d}}  
  \end{equation*}
\end{theorem}

\subsection{Statistical Estimation}

A key element in our algorithms for the sparse case is the use of least squares estimation
to reduce error. Below we present a bound on the error of least
squares estimation with respect to symmetric convex bodies. This analysis appears
to be standard in the statistics literature; a special case of it 
appears for example in~\cite{lse-stats}. 

\begin{figure*}[h]
  \begin{center}
  \includegraphics{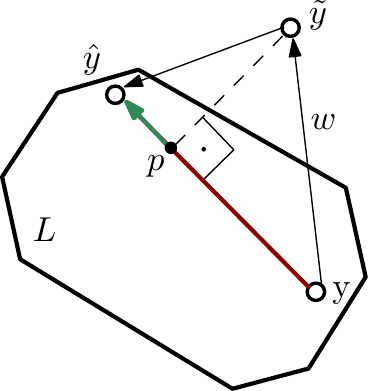}
  \caption{A schematic illustration of the proof of
    Lemma~\ref{lm:lse}. The vector $p-y$ is proportional in length to
    $\langle \hat{y} - y, w \rangle$ and the vector $\hat{y} - p$ is
    proportional in length to $\langle \hat{y} - y, \hat{y} -
    \tilde{y} \rangle$. As $\|w\| \geq \|\hat{y} - \|\tilde{y}\|$,
    $\|p-y\| \geq \|\hat{y} - p\|$. }
  \label{fig:lse}
  \end{center}
\end{figure*}

\begin{lemma}\label{lm:lse}
  Let $L \subseteq \R^d$ be  a symmetric convex body, and let $y \in
  L$ and $\tilde{y} = y + w$ for some $w \in \R^d$. Let, finally,
  $\hat{y} =$ $\arg \min_{\hat{y} \in L}{\|\hat{y} - \tilde{y}\|_2^2}.$ 
  We have
  $
    \|\hat{y} - y\|_2^2 \leq \min\{4\|w\|_2^2, 4\|w\|_{L^\circ}\}.
    $
\end{lemma}
\begin{proof}
  First we show the easier bound $\|\hat{y} - y\|_2 \leq 2\|w\|_2$,
  which follows by the triangle inequality:
  \begin{equation*}
    \|\hat{y} - y\|_2 \leq     \|\hat{y} - \tilde{y}\|_2 + \|\tilde{y}
    - y\|_2 \leq 2    \|\tilde{y} - y\|_2. 
  \end{equation*}

  The second bound is based on H\"{o}lder's inequality and the
  following simple but very useful fact, illustrated schematically in
  Figure~\ref{fig:lse}:
  \begin{align}
    \|\hat{y} - y\|_2^2 &= \langle \hat{y} - y, \tilde{y} - y \rangle + \langle
    \hat{y} - y, \hat{y} - \tilde{y} \rangle\notag\\
    &\leq 2\langle \hat{y} - y, \tilde{y} - y \rangle \label{eq:magic-ineq}.
  \end{align}
  \stocoption{}{The inequality (\ref{eq:magic-ineq}) follows from  }\stocoption{\junk{\begin{align*}
    \langle \hat{y} - y, \tilde{y} - y \rangle &= \|\tilde{y} - y\|_2^2 + \langle \hat{y} -
    \tilde{y}, \tilde{y} - y \rangle \\ &\geq \|\hat{y} - \tilde{y}\|_2^2 + \langle \hat{y} -
    \tilde{y}, \tilde{y} - y \rangle = \langle \hat{y} - \tilde{y}, \hat{y} - y\rangle.
  \end{align*}}}{\begin{equation*}
    \langle \hat{y} - y, \tilde{y} - y \rangle = \|\tilde{y} - y\|_2^2 + \langle \hat{y} -
    \tilde{y}, \tilde{y} - y \rangle \geq \|\hat{y} - \tilde{y}\|_2^2 + \langle \hat{y} -
    \tilde{y}, \tilde{y} - y \rangle = \langle \hat{y} - \tilde{y}, \hat{y} - y\rangle.
  \end{equation*}}
  Inequality (\ref{eq:magic-ineq}), $w = \tilde{y} - y$, and H\"{o}lder's inequality imply
  \begin{equation*}
    \|\hat{y} - y\|_2^2 \leq 2\langle \hat{y} - y, w \rangle \leq
    2\|\hat{y} - y\|_L\|w\|_{L^\circ} \leq 4\|w\|_{L^\circ},
  \end{equation*}
  which completes the proof.
\end{proof}

\subsection{Differential Privacy}

Following recent work in differential privacy, we model private data
as a database $D$ of $n$ rows, where each row of $D$ contains
information about an individual.  Formally, a database $D$ is a
multiset of size $n$ of elements of the universe $U = \{t_1, \ldots,
t_N\}$ of possible user types. Our algorithms take as input a
{histogram} $x \in \R^N$ of the database $D$, where the $i$-th
component $x_i$ of $x$ encodes the number of individuals in $D$ of
type $t_i$. Notice that in this histogram representation, we have
$\|x\|_1 = n$ when $D$ is a database of size $n$. Also, two
neighboring databases $D$ and $D'$ that differ in the presence or
absence of a single individual correspond to two histograms $x$ and
$x'$ satisfying $\|x - x'\|_1 = 1$.

Through most of this paper, we work under the notion of \emph{approximate
differential privacy}. The definition follows.

\begin{definition}[\cite{DMNS,odo}]
  A (randomized) algorithm $\mech$ with input domain $\mathbb{\R}^N$
  and output range $Y$ is \emph{$(\eps, \delta)$-differentially
    private} if for every $n$, every ${x}, {x}'$ with $\|{x} -
  {x}'\|_1 \leq 1$, and every measurable $S \subseteq Y$, $\mech$
  satisfies
  \begin{equation*}
    \Pr[\mech({x}) \in S] \leq e^\eps \Pr[\mech({x}') \in S] + \delta.
  \end{equation*}
\end{definition}
When $\delta = 0$, we are in the regime of \emph{pure differential
  privacy}.

An important basic property of differential privacy is that the
privacy guarantees degrade smoothly under composition and are not
affected by post-processing.

\begin{lemma}[\cite{DMNS,odo}]
  \label{lm:simple-composition}
  Let $\mech_1$ and $\mech_2$ satisfy $(\eps_1, \delta_1)$- and
  $(\eps_2, \delta_2)$-differential privacy, respectively. Then the
  algorithm which on input $\vec{x}$ outputs the tuple
  $(\mech_1(\vec{x}), \mech_2(\mech_1(\vec{x}), \vec{x}))$ satisfies
  $(\eps_1 + \eps_2, \delta_1 + \delta_2)$-differential privacy.
\end{lemma}

\subsubsection{Optimality for Linear Queries}

In this paper we study the necessary and sufficient error incurred by
differentially private algorithms for approximating \emph{linear queries}. A
set of $d$ linear queries is given by a $d \times N$ \emph{query
  matrix} or \emph{workload} $A$; the
exact answers to the queries on a histogram $x$ are given by the
$d$-dimensional vector $y = Ax$. 

We define error as total squared error. More precisely, for an
algorithm $\mech$ and a subset $X\subseteq \R^N$, we define
\begin{equation*}
  \err_\mech(A, X) \triangleq \sup_{x \in X} \E \|Ax - \mech(A, x)\|_2^2.
\end{equation*}

We also write $\err_\mech(A, nB_1^N)$ as $\err_\mech(A, n)$.
The optimal error achievable by any $(\eps, \delta)$-differentially
private algorithm for queries $A$ and databases of size up to $n$ is
\begin{equation*}
  \opt_{\eps, \delta}(A, n) \triangleq \inf_{\mech} \err_\mech(A, n),
\end{equation*}
where the infimum is taken over all $(\eps, \delta)$-differentially
private algorithms. When no restrictions are placed on the size $n$ of
the database, the appropriate notion of optimal error is $\opt_{\eps,
  \delta}(A) \triangleq \sup_{n}\opt_{\eps, \delta}(A, n)$. Similarly, for an
algorithm $\mech$, the error when database size is not bounded is
$\err_\mech(A) \triangleq \sup_{n} \err_\mech(A, n)$. A priori it is not clear
that these quantities are necessarily finite, but we will show that
this is the case.

In order to get tight dependence on the privacy parameter $\eps$
in our analyses, we will use the following relationship between
$\opt_{\eps, \delta}(A, n)$ and $\opt_{\eps', \delta'}(A, n)$. 

\begin{lemma}\label{lm:eps}
  For any $\eps$, any $\delta < 1$, any integer $k$ and for $\delta'
  \geq \frac{e^{k\eps} - 1}{e^\eps - 1}\delta$, 
  \begin{equation*}
    \opt_{\eps, \delta}(A, n) \geq k^2 \opt_{k\eps,
      \delta'}(A, n/k). 
  \end{equation*}
\end{lemma}
\stocoption{We defer the proof to the appendix.}{
\begin{proof}
  Let $\mech$ be an $(\eps, \delta)$-differentially private algorithm
  achieving $\opt_{\eps, \delta}(A, n)$. We will use $\mech$ as a
  black box to construct a $(k\eps, \delta')$-differentially private
  algorithm $\mech'$ which satisfies the error guarantee
  $\err_{\mech'}(A, n/k) \leq \frac{1}{k^2}\err_\mech(A, n)$. 

  The algorithm $\mech'$ on input $x$ satisfying $\|x\|_1 \leq n/k$
  outputs $\frac{1}{k}\mech(kx)$. We need to show that $\mech'$
  satisfies $(k\eps, \delta')$-differential privacy. Let $x$ and $x'$ be two
  neighboring inputs to $\mech'$, i.e. $\|x - x'\|_1 \leq 1$, and let
  $S$ be a measurable subset of the output $\mech'$. Denote $p_1 =
  \Pr[\mech'(x) \in S]$ and $p_2 = \Pr[\mech'(x') \in S]$. We need to
  show that $p_1 \leq e^{k\eps}p_2 + \delta'$. To that end, define
  $x_0 = kx$, $x_1 = kx + (x' - x)$, $x_2 = kx + 2(x'-x)$, $\ldots$,
  $x_k = kx'$. Applying the $(\eps, \delta)$-privacy guarantee of
  $\mech$ to each of the pairs of neighboring inputs $x_0, x_1$, $x_1,
  x_2$, $\ldots$, $x_{k-1}, x_k$ in sequence gives us
  \begin{equation*}
    p_1 \leq e^{k\eps}p_2 + (1 + e^\eps + \ldots + e^{(k-1)\eps})\delta =
    e^{k\eps}p_2 + \frac{e^{k\eps} - 1}{e^\eps - 1}\delta.
  \end{equation*}
  This finishes the proof of privacy for $\mech'$. It is
  straightforward to verify that $\err_{\mech'}(A, n/k) \leq
  \frac{1}{k^2}\err_\mech(A, n)$.  
\end{proof}
} Above, we state the error and optimal error definitions for
histograms $x$, which can be arbitrary real vectors. All our
algorithms work in this general setting. Recall, however, that the
histograms arising from our definition of databases are integer
vectors. Our lower bounds do hold against integer histograms as
well. Therefore, defining $\err$ and $\opt$ in terms of integer
histograms (i.e. taking $\err_\mech(A, n) \triangleq \err_\mech(A,
nB_1^N \cap \mathbb{N}^N)$) does not change the asymptotics of our
theorems.

\subsubsection{Gaussian Noise Mechanism}

A basic mechanism for achieving $(\eps, \delta)$-differential privacy
for linear queries is adding appropriately scaled independent Gaussian
noise to each query\stocoption{~\cite{sulq}.}{. This approach goes
  back to the work of Blum et al.~\cite{sulq}, predating the
  definition of differential privacy.} Next we define this basic
mechanism formally and give a privacy guarantee. The privacy analysis of the Gaussian mechanism in the context of $(\eps,
\delta)$-differential privacy was first given
in~\cite{odo}. \stocoption{The full proof is included in the appendix
  for completeness.}{We give the full proof here for completeness.}

\begin{lemma}\label{lm:gauss-mech-ind}
  Let $A = (a_i)_{i = 1}^N$ be a $d \times N$ matrix such that
  $\forall i: \|a_i\|_2 \leq \sigma$. Then a mechanism which on input
  $x \in \R^N$ outputs $Ax + w$, where $w \sim N(0, \sigma \frac{1 +
    \sqrt{2\ln(1/\delta)}}{\eps})^d$, satisfies $(\eps,
  \delta)$-differential privacy.
\end{lemma}
\stocoption{}{\begin{proof}
  Let $C = \frac{1+\sqrt{2\ln(1/\delta)}}{\eps}$ and let $p$ be the
  probability density function of $N(0, C\sigma)^d$. Let also $K =
  AB_1$, so $\|x - x'\|_1 \in B_1$ implies $A(x - x') \in K \subseteq
  B_2^d$. Define
  \begin{equation*}
    D_v(w) \triangleq \ln \frac{p(w)}{p(w + v)}.
  \end{equation*}
  We will prove that when $w \sim N(0, C\sigma)$, for all $v \in K$,
  $\Pr[|D_v(w)| > \eps] \leq \delta$. This suffices to prove $(\eps,
  \delta)$-differential privacy. Indeed, let the algorithm output $Ax
  + w$ and fix any $x'$ s.t. $\|x - x'\|_1 \leq 1$. Let $v = A(x-x')
  \in K$ and $S = \{w: |D_v(w)| > \eps\}$. For any measurable $T
  \subseteq \R^d$ we have
  \begin{align*}
    \Pr[Ax + w \in T] &= \Pr[w \in T - Ax] \\
    &= \int_{S \cap (T- Ax)}{p(w)\,\mathrm{d}w} + \int_{\bar{S} \cap
      (T-Ax)}{p(w)\,\mathrm{d}w}\\ 
    &\leq \delta + e^\eps\int_{\bar{S} \cap (T - Ax')}{p(w)\,\mathrm{d}w}\\
    &= \delta + e^\eps\Pr[w \in T - Ax']
    = \delta + e^\eps\Pr[Ax' + w \in T].
  \end{align*}
  
  We fix an arbitrary $v \in K$ and proceed to prove $|D_v(w)| \leq
  \eps$ with probability at least $1 - \delta$. We will first compute
  $\E D_v(w)$ and then apply a tail bound.  Recall that $p(w) \propto
  \exp(-\frac{1}{2C^2\sigma^2}\|w\|_2^2)$. Notice also that, since $v
  \in K$ can be written as $\sum_{i = 1}^N{\alpha_ia_i}$ where
  $\sum{|\alpha_i|} \leq 1$, we have $\|v\|_2 \leq \sigma$. Then we
  can write
  \begin{align*}
    \E D_v(w) &= \E \frac{\|v + w\|_2^2 - \|w\|_2^2}{2C^2\sigma^2}\\
    &=\E\frac{\|v\|^2 + 2v^Tw}{2C^2\sigma^2} \leq \frac{1}{2C^2}
  \end{align*}
  Note that to bound $|D_v(w)|$ we simply need to bound
  $\frac{1}{C^2\sigma^2}v^Tw$ from above and below. Since
  $\frac{1}{C^2\sigma^2}v^Tw \sim N(0, \frac{\|v\|}{C\sigma})$, we can
  apply a Chernoff bound and we get
  \begin{equation*}
    \Pr\left[|v^Tw| > \frac{1}{C}\sqrt{2\ln(1/\delta)}\right] \leq \delta.
  \end{equation*}
  Therefore, with probability $1-\delta$,
  \begin{equation*}
    \frac{1/2C - \sqrt{2\ln(1/\delta)}}{C} \leq D_v(w) \leq \frac{1/2C
      + \sqrt{2\ln(1/\delta)}}{C}. 
  \end{equation*}
  Substituting $C \geq \frac{1 + \sqrt{2\ln(1/\delta)}}{\eps}$
  completes the proof. 
\end{proof}}

The following corollary is a useful geometric generalization of
Lemma~\ref{lm:gauss-mech-ind}. 
\begin{corollary}\label{cor:gauss-incl}
  Let $A = (a_i)_{i=1}^N$ be a $d \times N$ matrix of rank $d$ and let
  $K = \sym\{a_1, \ldots, a_N\}$. Let $E = FB_2^d$ ($F$ is a linear map)
  be an ellipsoid containing $K$. Then a mechanism that outputs $Ax +
  Fw$ where $w \sim N(0, \frac{1 +
    \sqrt{2\ln(1/\delta)}}{\eps})^d$ satisfies $(\eps,
  \delta)$-differential privacy.
\end{corollary}
\stocoption{}{\begin{proof}
  Since $K$ is full dimensional (by $\rank A = d$) and $E$ contains
  $K$, $E$ is full dimensional as well, and, therefore, $F$ is an
  invertible linear map. Define $G = F^{-1}A$. For each column $g_i$
  of $G$, we have $\|g_i\|_2 \leq 1$. Therefore, by
  Lemma~\ref{lm:gauss-mech-ind}, a mechanism that outputs $Gx + w$
  (where $w$ is distributed as in the statement of the corollary)
  satisfies $(\eps, \delta)$-differential privacy. Therefore, $FGx +
  Fw = Ax + Fw$ is $(\eps, \delta)$-differentially private by the
  post-processing property of differential privacy. 
\end{proof}
}
We present a composition theorem, specific to composing
Gaussian noise mechanisms. We note that a similar composition result
in a much more general setting but with slightly inferior dependence
on the parameters is proven in~\cite{DworkRV10}. 
\begin{corollary}\label{cor:gauss-composition}
  Let $\mathcal{V}_1, \ldots, \mathcal{V}_k$ be vector spaces of
  respective dimensions $d_1, \ldots, d_k$, such that $\forall i \leq
  k-1$, $\mathcal{V}_{i+1} \subseteq \mathcal{V}_i^\perp$ and $d_1 +
  \ldots + d_k = d$ . Let $A = (a_i)_{i=1}^N$ be a $d \times N$ matrix
  of rank $d$ and let $K = \sym(a_1, \ldots, a_N)$. Let $\Pi_i$ be the
  projection matrix for $\mathcal{V}_i$ and let $E_i = F_iB_2^{d_i}
  \subseteq \mathcal{V}_i$ be an ellipsoid such that $\Pi_i K \subseteq
  E_i$. Then the mechanism that outputs $Ax +
  \sqrt{k}\sum_{i = 1}^k{F_iw_i}$ where for each $i$, $w_i
  \sim N(0,  \frac{1 + \sqrt{2\ln(1/\delta)}}{\eps})^{d_i}$,
  satisfies $(\eps, \delta)$-differential privacy.
\end{corollary}
\stocoption{}{\begin{proof}
  Let $c(\eps, \delta) = \frac{1+\sqrt{2\ln(1/\delta)}}{\eps}$. Since the random
  variables $F_1w_1, \ldots, F_kw_k$ are pairwise independent Gaussian
  random variables, and $F_iw_i$ has covariance matrix $c(\eps, \delta)^2 F_iF_i^T$, we have that $w = \sqrt{k}\sum_{i = 1}^k{F_iw_i}$ is a
  Gaussian random variable with covariance $c(\eps, \delta)^2G$, whee $G = k
  \sum_{i = 1}^k{F_iF_i^T}$. By Corollary~\ref{cor:gauss-incl}, it is
  sufficient to show that the ellipsoid $E = GB_2^d$ contains $K$. By
  convex duality, this is equivalent to showing $E^\circ \subseteq
  K^\circ$, which is in turn equivalent to $\forall x: \|x\|_{K^\circ} \leq
  \|x\|_{E^{\circ}}$. Recalling that $\|x\|^2_{E^{\circ}} = x^TGG^Tx$
  and $\|x\|_{K^\circ} = \max_{y \in K}{\langle y, x\rangle} = \max_{j
    = 1}^N{\langle a_j, x \rangle}$,  we need to establish
  \begin{equation}
    \label{eq:inclusion-duality}
    \forall x \in \R^d, \forall j \in [N]: \langle a_j, x \rangle^2
    \leq x^TGG^Tx. 
  \end{equation}
  We proceed by establishing (\ref{eq:inclusion-duality}). Since for
  all $i$, $\Pi_i K \subseteq E_i$, by duality and the same reasoning
  as above, we have that for all $i$ and $j$, $\langle \Pi_i a_j, x
  \rangle^2 \leq x^TF_iF_i^Tx$. Therefore, by the Cauchy-Schwarz
  inequality,
  \begin{align*}
    \langle a_j, x \rangle^2 &= \left(\sum_{i = 1}^k{\langle \Pi_i
        a_j, x \rangle}\right)^2\\
    &\leq k\sum_{i = 1}^k{\langle \Pi_i  a_j, x \rangle^2}\\
    &\leq k\sum_{i = 1}^k{x^TF_iF^Tx} = x^TGG^Tx. 
  \end{align*}
  This completes the proof.
\end{proof}}

\subsubsection{Noise Lower Bounds}
We will make extensive use of a lower bound on the noise complexity of
$(\eps, \delta)$-differentially private mechanisms in terms of
combinatorial discrepancy. First we need to define the notion of
hereditary $\alpha$-discrepancy:
\begin{equation*}
  \herdisc_\alpha(A, n) \triangleq \max_{S \subseteq [N]: |S| \leq n}{\min_{\substack{x \in \{-1,0,
      +1\}^S\\\|x\|_1 \geq \alpha |S|}}{\|(A|_S)x}\|_2}.
\end{equation*}
We denote $\herdisc(A) \triangleq \max_n \herdisc_1(A, n)$. An equivalent
notation is $\herdisc^{\ell_2}(A)$. When the $\ell_2$ norm is
substituted with $\ell_\infty$, we have the classical notion hereditary
discrepancy, here denoted $\herdisc^{\ell_\infty}(A)$. 

Next we present the lower bound, which is a simple extension of the
discrepancy lower bound on noise recently proved by Muthukrishnan and
Nikolov~\cite{MNstoc}. 
\begin{theorem}[\cite{MNstoc}]
  \label{thm:disc-lb}
  Let $A$ be an $d \times N$ real matrix. For any constant $\alpha$
  and sufficiently small constant $\eps \leq \eps(\alpha)$ and $\delta
  \leq \delta(\alpha)$,
  \begin{equation*}
    \opt_{\eps, \delta}(A, n) = \Omega(1) \herdisc_\alpha(A, n)^2. 
  \end{equation*}
\end{theorem}

We further develop two lower bounds for $\herdisc_\alpha(A,n)$ which are more convenient to work with. The first lower bound is by using spectral
techniques. Observe first that, since the $\ell_2$-norm of any vector
does not increase under projection, we have $\herdisc_\alpha(A, n)
\geq \herdisc_\alpha(\Pi A, n)$ for any projection matrix $\Pi$. 
Furthermore, recall that for a matrix $M$, $\sigma_{\min}(M) =
\min_{x:\|x\|_2 = 1}{\|Mx\|_2}$. For any $x \in \{-1, 0, 1\}^S$
satisfying $\|x\|_1 \geq \alpha |S|$, we have $\|x\|_2^2 \geq
\alpha |S|$. Therefore, 
\begin{equation}\label{eq:disc-speclb}
  \herdisc_\alpha(A, n)^2 \geq \max_{S \subseteq [N]: |S| \leq
    n}{\alpha |S| \sigma_{\min}^2(A|_S)}. 
\end{equation}

Let's define
\begin{equation*}
\specLB(A, n) \triangleq \max_{\substack{S \subseteq [N]\\|S| = k \leq
    n}} \max_{\Pi \in \mathcal{P}_k}{k \sigma_{\min}^2(\Pi A|_S)}.
\end{equation*}
Substituting (\ref{eq:disc-speclb}) into Theorem~\ref{thm:disc-lb}, we
have that there exist constants $c_1$ and $c_2$ such that
\begin{equation}
  \label{eq:speclb}
  \opt_{c_1, c_2}(A, n) = \Omega(1) \cdot \specLB(A, n). 
\end{equation}
For the remainder of this paper we fix some constants $c_1$ and $c_2$
for which (\ref{eq:speclb}) holds. Similarly to the notation for
$\opt$, we will also sometimes denote $\specLB(A) = \max_n \specLB(A,
n)$.  We will use primarily the spectral lower bound~(\ref{eq:speclb}) for arguing the optimality of our algorithms. 

\stocoption{A more powerful spectral lower bound follows from
  relating $\herdisc_\alpha$ and $\herdisc_1$, and from a direct extension
  of the work of Lov\'{a}sz, Spencer, and
  Vesztergombi~\cite{lovasz1986discrepancy}. Define
  $$
  \detlb(A, n) \triangleq \max_{k      \leq n}    \max_{\substack{\Pi \in \mathcal{P}_k\\S \subseteq [N]: |S| =        k}}{|S|\cdot |\det(\Pi A|_S)|^{2/|S|}}.
  $$
  Then we have the lower bound:
  }{
To show the small gap between the approximate and pure privacy (Theorem~\ref{thm:epsdelgap}), we next develop a determinant based lower bound.
We first switch from $\herdisc_\alpha$ to the classical notion of hereditary
discrepancy, equivalent to $\herdisc_1$, by observing
the following relation between
$\herdisc_\alpha$ and $\herdisc_1$ from~\cite{MNstoc}:
\begin{equation}
  \label{eq:alpha2one}
  \herdisc_{1}(A, n) \leq \frac{\log n}{\log
    \frac{1}{1-\alpha}}\herdisc_\alpha(A, n) 
\end{equation}

We then use an extension of the classical determinant lower bound for
hereditary discrepancy, due to Lov\'{a}sz, Spencer, and Vesztergombi.
\begin{theorem}[\cite{lovasz1986discrepancy}]
  For any real $d \times N$ matrix ${A}$,
  \begin{equation*}
   (\herdisc_1({A}, n))^2 \geq \Omega(1) \cdot \detlb(A, n) \triangleq \max_{k
      \leq n}
    \max_{\substack{\Pi \in \mathcal{P}_k\\S \subseteq [N]: |S| = k}}{k\cdot |\det(\Pi A|_S)|^{2/k}}.
  \end{equation*}
\end{theorem}
\begin{proof}
  The proof proceeds in two steps: showing that a quantity known as
  linear discrepancy is at most a constant factor larger than
  hereditary discrepancy; lower bounding linear discrepancy in terms
  of $\detlb$. Most of the proof can be adapted with little
  modification from proofs of the lower bound on $\ell_\infty$
  discrepancy in~\cite{lovasz1986discrepancy}. Here we follow the
  exposition in~\cite{Chazelle}.

  Let us first define linear discrepancy. For a $d\times k$ matrix
  $M$ and $c \in [-1, 1]^k$, let  $\disc^c(M)$ be defined as
  \begin{equation*}
    \disc^c(M) \triangleq \min_{x \in \{-1, 1\}^k}{\|Mx - Mc\|_2}
  \end{equation*}
  The linear discrepancy of $M$ is then defined as $\lindisc(M) \triangleq
  \max_{c \in [-1, 1]^k}{\disc^c(M)}$. We claim that for any $M$, 
  \begin{equation}
    \label{eq:lindisc-herdisc}
    \lindisc(M) \leq 2 \herdisc(M).
  \end{equation}
  The bound (\ref{eq:lindisc-herdisc}) is proven for the more common
  $\ell_\infty$ variants of $\lindisc$ and $\herdisc$
  in~\cite{Chazelle}, but the proof can be seen to apply without
  modification to discrepancy defined with respect to any norm. For
  completeness, we give the full argument here. For $c \in \{-1, 0,
  1\}^k$, we have $\disc^c(M) \leq \herdisc(M)$. Call a vector $c \in
  [-1, 1]^k$ $q$-integral if any coordinate $c_i$ of $c$ can be
  written as $\sum_{j = 0}^q{b_j 2^{-j}}$ where $(b_j)_{j=0}^q\in \{0,
  1\}^{q+1}$. In other words, $c$ is $q$-integral if the binary
  expansion of any coordinate of $c$ is identically zero after the
  $q$-th digit. The bound (\ref{eq:lindisc-herdisc}) holds for
  $0$-integral $c$, and we prove that it holds for $q$-integral $c$ by
  induction on $q$. For the induction step, assume that the bound
  holds for any $(q-1)$-integral $c'$ and let $c$ be
  $q$-integral. Define $s \in \{-1, 1\}^n$ by setting $s_i = -1$ if
  $c_i \geq 0$ and $s_i = 1$ otherwise. Then $c' = 2c + s \in [-1,
  1]^k$ is $(q-1)$-integral, and, by the induction hypothesis, there
  exists $x_1 \in \{-1, 1\}^k$ such that $\|Mx_1 - Mc'\|_2 \leq
  2\herdisc(M)$. Let $c'' = \frac{x_1 - s}{2} \in \{-1, 0,
  1\}^k$. Dividing by 2 and rearranging, we have $\|Mc'' - Mc\|_2 \leq
  \herdisc(M)$. The vector $c''$ is $0$-integral, and therefore there
  exists some $x_2 \in \{-1, 1\}^k$ such that $\|Mx_2 - Mc''\|_2 \leq
  \herdisc(M)$. By the triangle inequality, $\|Mx_2 - Mc\|_2 \leq
  2\herdisc(M)$, and this completes the inductive step.

  We complete the proof of the theorem by proving a lower bound on
  $\lindisc$ in terms of $\detlb$. We note that a similar lower bound
  can be proved for any variant of $\lindisc$ defined in terms of any
  norm. The exact lower bound will depend on the volume radius of the
  unit ball of the norm. Since the proof of (\ref{eq:lindisc-herdisc})
  also works for any norm, we get a determinant lower bound for
  hereditary discrepancy defined in terms of any norm as well. 
  
  We show that for any $d \times k$ matrix $M$
  \begin{equation}
    \label{eq:detlb-lindisc}
    \lindisc(M) = \Omega(1) \max_{\Pi \in \mathcal{P}_k}\sqrt{k}|\det \Pi M|^{1/k}.
  \end{equation}
  Letting $k$ range over $[n]$, $M$ range over all $d \times k$
  submatrices of $A$, and applying the bounds
  (\ref{eq:lindisc-herdisc}) and (\ref{eq:detlb-lindisc}) implies the
  theorem.

  We proceed to prove (\ref{eq:detlb-lindisc}). Note that if $\rank(M)
  < k$, (\ref{eq:detlb-lindisc}) is trivially true; therefore, we may
  assume that $\rank(M) = k$. Note also that without loss of
  generality we can take $\Pi$ to be the orthogonal projection onto
  the range of $M$, since this is the projection operator that
  maximizes $|\det \Pi M|$.  Let $E$ be the ellipsoid $E = \{x:
  \|Mx\|_2^2 \leq 1\}$. The inequality $\lindisc(M) \leq D$ is
  equivalent to
  \begin{equation}
    \label{eq:cover}
    [-1, 1]^k \subseteq \bigcup_{x \in \{-1, 1\}^k}{D\cdot E + x}.
  \end{equation}
  Thus $2^k = \vol([-1,1]^d) \leq 2^k \vol(D\cdot E)$, and therefore
  $D^k \geq \frac{1}{\vol(E)}$. On the other hand, the volume of $E$
  is equal to
  $$
  \vol(E) = \frac{\vol(B_2)}{|\det(M^TM)|^{1/2}} =
  \frac{\vol(B_2)}{|\det(\Pi M)|}. 
  $$
  Applying the standard estimate $\vol(B_2^k)^{1/k} =
  \Theta(k^{-1/2})$ completes the proof. 
\end{proof}

By the determinant lower bound, and (\ref{eq:alpha2one}), we get our
determinant lower bound on the noise necessary for privacy.} For some constant $c_1,c_2>0$,
\begin{equation}
  \label{eq:detlb}
  \opt_{c_1, c_2}(A, n) = \Omega\left(\frac{1}{\log^2 n}\right) \detlb(A, n). 
\end{equation}

Finally, we recall the stronger volume lower bound against $(\eps,
0)$-differential privacy from~\cite{HardtT10, 12vollb}. This lower bound is
nearly optimal for $(\eps, 0)$-differential privacy, but does not hold
for $(\eps, \delta)$-differential privacy when $\delta$ is
$2^{-o(d)}$. 

\begin{theorem}[\cite{HardtT10,12vollb}]\label{thm:htvollb}
  For any $d \times N$ real matrix $A = (a_i)_{i=1}^N$, 
  \begin{equation}\label{eq:vollb}
    \opt_{\eps, 0}(A, \frac{d}{\eps}) \geq \vollb(A, \eps) \triangleq
    \max_{k=1}^d \max_{\Pi \in  \mathcal{P}_k}{\frac{k^2}{\eps^2}
      \vrad(\Pi K)^2},   
  \end{equation}
  where $K = \sym\{a_i\}_{i = 1}^d$. 

  Furthermore, there exists an efficient mechanism $\mech_K$ (the
  \emph{generalized $K$-norm mechanism}) which is $(\eps,
  0)$-differentially private and satisfies $\err_{\mech_K}(A) =
  O(\log^3 d) \vollb(A, \eps)$. 
\end{theorem}
\label{sect:prelims}

\section{Algorithms for Approximate Privacy}\label{sect:apx}

In this section we present our main results: efficient nearly optimal
algorithms for approximate privacy in the cases of dense databases ($n
> d/\eps$) and sparse databases ($n = o(d/\eps)$). Both algorithms
rely on recursively computing an orthonormal basis for $\R^d$, based
on the minimum volume enclosing ellipsoid of the columns of the query
matrix $A$. We first present the algorithm for computing this basis,
together with a property essential for the analyses of the two
algorithms presented next. 

\junk{ }

\subsection{The Base Decomposition Algorithm}
We first present an algorithm (Algorithm~\ref{alg:base}) that, given a
matrix $A \in \R^{d \times N}$, computes a set of orthonormal
matrices $U_1, \ldots, U_k$, where $k \leq \lceil 1 + \log d
\rceil$. For each $i \neq j$, $U_i^T U_j = 0$, and the union of the
columns of $U_1, \ldots, U_k$ forms an orthonormal basis for
$\R^d$. Thus, Algorithm~\ref{alg:base} computes a basis for $\R^d$,
and partitions (``decomposes'') it into $k = O(\log d)$ bases of
mutually orthogonal subspaces.\stocoption{The decomposition allows us
  to use a simple mechanism in one subspace and match its error by a
  lower bound for the a set of orthogonal subspaces.}{ This set of bases also induces a
decomposition of $A$ into $A = A_1 + \ldots + A_k$, where $A_i =
U_iU_i^TA$. The base decomposition of Algorithm~\ref{alg:base} is
essential to both our dense case and sparse case
algorithms. Intuitively, for both cases we can show that the error of
a simple mechanism applied to $A_i$ can be matched by an error lower
bound for $A_{i+1} + \ldots A_{k}$. The error lower bounds are based
on the spectral lower bound $\specLB$ on discrepancy; the geometric
properties of the minimum enclosing ellipsoid of a convex body
together with the restricted invertibility principle of Bourgain and
Tzafriri are key in deriving the lower bounds.}

\begin{algorithm}
  \caption{\textsc{Base Decomposition}} \label{alg:base}
  \begin{algorithmic}
    \REQUIRE $A = (a_i)_{i = 1}^N \in \R^{d \times N}$ ($\rank A =
    d$); 

    \STATE Compute $E = FB_2^d$, the minimum volume enclosing
    ellipsoid of $K = AB_1$; 

    \STATE Let $(u_i)_{i = 1}^d$ be the (left) singular vectors of $F$
    corresponding to singular values $\sigma_1 \geq \ldots \geq
    \sigma_d$;

    \IF{$d = 1$}
    \STATE \textbf{Output} $U_1 = u_1$. 
    \ELSE
    \STATE Let $U_1 = (u_i)_{i > d/2}$ and $V = (u_i)_{i \leq d/2}$;
    \STATE Recursively compute a base decomposition $V_2, \ldots, V_k$
    of $V^TA$ ($k \leq \lceil 1 + \log d\rceil$ is the depth of the
    recursion);
    \STATE For each $i>1$, let $U_i = VV_i$;
    \STATE \textbf{Output} $\{U_1, \ldots, U_k\}$.
    \ENDIF
  \end{algorithmic}
\end{algorithm}

The next lemma captures the technical property of the decomposition of
Algorithm~\ref{alg:base} that allows us to prove matching upper and
lower error bounds for our dense and sparse case algorithms.

\begin{lemma}\label{lm:base}
  Let $d_i$ be the dimension of the span of $U_i$. Furthermore, for $i
  < k$, let $W_i = \sum_{j > i}{U_j}$, and let $W_k = U_k$. For every
  $i \leq k$, there exists a set $S_i \subseteq [N]$, such that $|S_i| =
  \Omega(d_i)$ and $\sigma_{\min}^2(W_iW_i^TA|_{S_i}) = \Omega(1) \max_{j
    =1}^N{\|U_i^Ta_j\|_2^2}$.
\end{lemma}
\begin{proof}
  Let us, for ease of notation, assume that $d$ is a power of $2$. We
  prove that there exists a set $S$, $|S| = \Omega(d)$, such that
  \begin{equation}
    \label{eq:base}
    \sigma_{\min}^2(VV^TA|_S) = \Omega(1) \max_{j \in N}{\|U_1^Ta_j\|_2^2}.
  \end{equation}
  This proves the lemma for $i=1$ (substitute $W_i = V$ and $d_i =
  d/2$). Applying (\ref{eq:base}) inductively to $V^TA$ establishes
  the lemma for all $i < k$. In case $i = k$, observe that $d_k = 1$
  and we only need to show that the matrix $U_kU_k^TA$ has a column
  with $\ell_2$ norm at least $\max_{j \in N}{\|U_k^Ta_j\|_2}$. This is
  trivially true since each column of $U_kU_k^TA$ has the same
  $\ell_2$ norm as the corresponding column of $U_k^TA$. 

  By applying an appropriate unitary transformation to the columns of
  $A$, we may assume that the major axes of $E$ are co-linear with the
  standard basis vectors of $\R^d$, and, therefore, $F$ is a diagonal
  matrix with $F_{ii} = \sigma_i$. This transformation comes without
  loss of generality, since it applies a unitary transformation to the
  columns of $A$ and $V$ and does not affect the singular values of
  any matrix $VV^TA|_S$ for any $S \subseteq [N]$. For the rest of the
  proof we assume that $F$ is diagonal.

  Since $F$ is diagonal, $u_i$ is equal to $e_i$, the $i$-th standard
  basis vector. Therefore $U_1$ is diagonal and equal to the
  projection onto $e_{d/2 + 1}, \ldots, e_d$, and $V$ is also diagonal
  and equal to the projection onto $e_1, \ldots, e_{d/2}$. Consider $L
  = F^{-1}K = F^{-1}AB_1^d$ (recall that we assumed that $\rank A = d$
  and therefore $F$ is non-singular). Since the minimum enclosing
  ellipsoid of $K$ is $E = FB_2^d$, we have that the minimum enclosing
  ellipsoid of $L$ is $B_2^d$. Let $T = VV^T$ be the projection
  onto $e_1, \ldots, e_{d/2}$. Then, by Theorem~\ref{thm:bt}, and
  because $\|T\|_F^2 = d/2$, we have that there exists a set $S$ of
  size $|S| = \Omega(d)$ such that $\sigma_{\min}^2(TF^{-1}A|_S) =
  \Omega(1)$. We chose $F$, and therefore $F^{-1}$, as well as $T$ to
  be diagonal matrices, so they all commute. Then, since $T$ is a
  projection matrix,
\stocoption{\begin{align}
    \sigma_{\min}^2&(VV^TA|_S) = \sigma_{\min}^2(TA|_S)=
    \sigma_{\min}^2(F T F^{-1}A|_S) \notag\\ 
    &= \sigma_{\min}^2(F T T F^{-1}A|_S) 
    \geq \sigma_{\min}^2(FT) \sigma_{\min}^2(TF^{-1}A|_S)\notag\\ 
    &=  \Omega(\sigma_{d/2}^2).\label{eq:singval-lb}
  \end{align}}{\begin{align}
    \sigma_{\min}^2(VV^TA|_S) &= \sigma_{\min}^2(TA|_S)=
    \sigma_{\min}^2(F T F^{-1}A|_S) =
    \sigma_{\min}^2(F T T F^{-1}A|_S) \notag\\
    &\geq \sigma_{\min}^2(FT) \sigma_{\min}^2(TF^{-1}A|_S) =
    \Omega(\sigma_{d/2}^2).\label{eq:singval-lb}
  \end{align}}

  Observe that, since $K  \subseteq E$, we have 
  \begin{equation*}
    U_1^TK \subseteq U_1^T E = U_1^T F B_2^d \subseteq \sigma_{d/2 +
      1}B_2^d. 
  \end{equation*}
  Therefore, $\max_{j = 1}^N {\|U_1^T a_j\|_2^2} \leq \sigma_{d/2+ 1}^2
  \leq \sigma_{d/2}^2$. Substituting for $\sigma_{d/2}^2$ into
  (\ref{eq:singval-lb}) completes the proof.
\end{proof}

\subsection{The Dense Case: Correlated Gaussian Noise}\label{sec:densepsdel}

Our first result is an efficient algorithm whose expected error
matches the spectral lower bound $\specLB$ up to polylogarithmic
factors and is therefore nearly optimal. This proves Theorem~\ref{thm:epsdeldense}. The algorithm adds
correlated unbiased Gaussian noise to the exact answer $Ax$. The noise
distribution is computed based on the decomposition algorithm of the
previous subsection.

\begin{algorithm}
  \caption{\textsc{Gaussian Noise Mechanism}}\label{alg:gaussnoise}
  \begin{algorithmic}
    \REQUIRE \emph{(Public)}: query matrix $A = (a_i)_{i = 1}^N \in \R^{d
      \times N}$ ($\rank A = d$); 
    \REQUIRE \emph{(Private)}: database $x \in \R^N$

    \STATE Let $U_1, \ldots, U_k$ be base decomposition computed by
    Algorithm~\ref{alg:base} on input $A$, where $U_i$ is an
    orthonormal basis for a space of dimension $d_i$;

    \STATE Let $c(\eps, \delta) =
    \frac{1+\sqrt{2\ln(1/\delta)}}{\eps}$.

    \STATE For each $i$, let $r_i = \max_{j = 1}^N{\|U_i^Ta_j\|_2}$;

    \STATE For each $i$, sample $w_i \sim N(0, c(\eps, \delta))^{d_i}$;
    \STATE \textbf{Output} $Ax + \sqrt{k}\sum_{i = 1}^k{r_iU_iw_i}$. 
  \end{algorithmic}
\end{algorithm}

\begin{theorem} \label{thm:dense} Let $\mech_g(A, x)$ be the output of
  Algorithm~\ref{alg:gaussnoise} on input a $d \times N$ query matrix
  $A$ and private input $x$.  $\mech_g(A, x)$ is $(\eps,
  \delta)$-differentially private and for all small enough $\eps$ and all
  $\delta$ small enough with respect to $\eps$ satisfies
\stocoption{\begin{align*}
    \err_{\mech_g}(A) &= O(\log^2 d\log 1/\delta) \opt_{\eps,
      \delta}(A, \frac{d}{\eps})\\ &= O(\log^2 d\log 1/\delta) \opt_{\eps,
      \delta}(A). 
  \end{align*}}{\begin{equation*}
    \err_{\mech_g}(A) = O(\log^2 d\log 1/\delta) \opt_{\eps,
      \delta}(A, \frac{d}{\eps}) = O(\log^2 d\log 1/\delta) \opt_{\eps,
      \delta}(A). 
  \end{equation*}}
\end{theorem}

Notice that even though we assume $\rank A = d$, this is without loss
of generality: if $\rank A = r < d$, we can compute an orthonormal
basis $V$ for the range of $A$ and apply the algorithm to $A' = V^T A$
to compute an approximation $\tilde{z}$ to $A'x$. We have that $VA'x =
VV^TAx = Ax$ since $VV^T$ is a projection onto the range of $A$ and
$Ax$ belongs to that range. Then $\tilde{y} = V\tilde{z}$ gives us an
approximation of $Ax$ satisfying $\|\tilde{y} - Ax\|_2 = \|\tilde{z} -
A'x\|_2$.

We start the proof of Theorem~\ref{thm:dense} with the privacy
analysis. For ease of notation, we assume throughout the analysis that
$d$ is a power of 2. 

\begin{lemma}\label{lm:dense-priv}
  $\mech_g(A,x)$ satisfies $(\eps, \delta)$-differential privacy. 
\end{lemma}
\begin{proof}
  The lemma follows from Corollary~\ref{cor:gauss-composition}.  Next
  we describe in detail why the corollary applies. 

  Let $U_1, \ldots, U_k$ be the base decomposition computed by
  Algorithm~\ref{alg:base} on input $A$. Let $\mathcal{V}_i$ be the
  subspace spanned by the columns of $U_i$ and let $d_i$ be the
  dimension of $\mathcal{V}_i$. The projection matrix onto
  $\mathcal{V}_i$ is $U_iU_i^T$. Let $E_i$ be the ellipsoid $U_i
  (r_iB_2^{d_i}) = F_i B_2^{d_i}$ ($F_i$ is $r_iU_i$). By the
  definition of $r_i$, $U_i^TK \subseteq r_iB_2^{d_i}$, and therefore,
  $\Pi_i K \subseteq E_i$. $\mech_{g}(A, x)$ is distributed
  identically to $Ax + \sqrt{k}\sum_{i=1}^k{F_iw_{i}}$.  Therefore, by
  Corollary~\ref{cor:gauss-composition}, $\mech_g(A, x)$ satisfies
  $(\eps, \delta)$-differential privacy.
\end{proof}

\begin{lemma}\label{lm:dense-util}
  For all small enough $\eps$ and all $\delta$ small enough with
  respect to $\eps$, for all $i$,
  \begin{equation*}
  \E\|r_i U_iw_i\|_2^2 \leq O(\log \frac{1}{\delta})\opt_{\eps,
    \delta}(A, \frac{d_i}{\eps}),
  \end{equation*}
  where $r_i$, $U_i$ and $w_i$ are as defined in
  Algorithm~\ref{alg:gaussnoise}.
\end{lemma}
\begin{proof}
  \junk{
  Let us use the notation from Algorithm~\ref{alg:base} for $E$, $F$,
  $\sigma_i$ $U_1$ and $V$. 

  We show that 
  \begin{equation}
    \label{eq:U1-dense-ub}
    \E\|r_1U_1w_1\|_2^2 \leq O(\log\frac{1}{\delta})\opt_{\eps, \delta}(A).
  \end{equation}
  Since $\opt_{\eps, \delta}(V^TA) \leq \opt_{\eps, \delta}(A)$,
  applying (\ref{eq:U1-dense-ub}) recursively to $V^TA$ implies the
  statement of the lemma for all $i$.
  }

  To upper bound $\E\|r_iU_iw_i\|_2^2$, observe that, since the
  columns of $U_i$ are $d_i$ pairwise orthogonal unit vectors, $\|r_i
  U_i w_i\|_2^2 = r_i^2 \|w_i\|_2^2$.  Therefore, it follows that
  \begin{equation}
    \E \|r_iU_iw_i\|_2^2 = d_i c(\eps,   \delta)^2 r_i^2 = O(\log
    (1/\delta) \frac{1}{\eps^2}dr_i^2).\label{eq:dense-ub} 
  \end{equation}
  By (\ref{eq:dense-ub}), it is enough to lower bound $\opt_{\eps,
    \delta}(A)$ by $\Omega(\frac{1}{\eps^2} dr_i^2)$. As a
  first step we lower bound $\specLB(A)$. Then the lower bound on
  $\opt_{\eps, \delta}$ will follow from (\ref{eq:speclb}) and
  Lemma~\ref{lm:eps}. 

  To lower bound $\specLB(A)$, we invoke Lemma~\ref{lm:base}. It
  follows from the lemma that for every $i$ there exists a projection matrix $\Pi_i =
  W_iW_i^T$ and a set $S_i$ such that $\sigma_{\min}^2(\Pi_iA|_{S_i})
  = \Omega(r_i^2)$, and, furthermore, $|S_i| = \Omega(d_i)$.  
  Substituting into the the definition of $\specLB(A, d_i)$, we have that
  for all $i$.
  \begin{equation*}
    \specLB(A, d_i) \geq |S_i| \sigma_{\min}^2(\Pi_iA|_{S_i}) =
    \Omega(d_ir_i^2). 
  \end{equation*}
  Therefore, by (\ref{eq:speclb}), $\opt_{c_1, c_2}(A, d_i) =
  \Omega(d_ir_i^2)$ for all $i$. Finally, by Lemma~\ref{lm:eps}, there
  exists a small enough $\delta = \delta(\eps)$, for which
  $\opt_{\eps, \delta}(A, \frac{d_i}{\eps}) = \Omega(\frac{1}{\eps^2}d_ir_i^2)$,
  and this completes the proof.
\end{proof}

\begin{proofof}{Theorem~\ref{thm:dense}}
  The proof of the theorem follows from Lemma~\ref{lm:dense-priv} and
  Lemma~\ref{lm:dense-util}. The privacy guarantee is direct from
  Lemma~\ref{lm:dense-priv}. Next we prove that the error of $\mech_g$
  is near optimal.

  Let $w$ be the noise vector generated by
  Algorithm~\ref{alg:gaussnoise} so that $w= \sqrt{k}\sum_{i =1}^{1 + \log
    d}{r_iU_iw_i}$. We have $\err_{\mech_g}(A) = \E \|w\|_2^2$. We
  proceed to upper bound this quantity in terms of
  $\opt_{\eps,\delta}(A)$.

  By Lemma~\ref{lm:dense-util}, for each $w_i$,  
  $\E \|r_iU_iw_i\|_2^2 = O(\log 1/\delta)\opt_{\eps, \delta}(A)$.
  Since $U_i^TU_j = 0$ for all $i \neq j$, and $k = O(\log d)$, we can
  bound $\E \|w\|_2^2$ \stocoption{by}{as}
\stocoption{\begin{align*}
k \sum_{i = 1}^{1+\log d}&{E\|r_iU_iw_i\|_2^2} =
    O(\log d\log 1/\delta)\sum_{i = 1}^k{\opt_{\eps, \delta}(A,
      \frac{d_i}{\eps})}\\
     &= O(\log^2 d\log 1/\delta)\opt_{\eps, \delta}(A,\frac{d}{\eps})
  \end{align*}}{
  \begin{equation*}
    \E\|w\|_2^2 = k \sum_{i = 1}^{1+\log d}{E\|r_iU_iw_i\|_2^2} =
    O(\log d\log 1/\delta)\sum_{i = 1}^k{\opt_{\eps, \delta}(A,
      \frac{d_i}{\eps})}
      = O(\log^2 d\log 1/\delta)\opt_{\eps, \delta}(A,\frac{d}{\eps})
  \end{equation*}}
  This  completes the proof.   
\end{proofof}

\subsection{The Sparse Case: Least Squares Estimation}
\label{sect:sparse}

\junk{For each query matrix $A$, Algorithm~\ref{alg:gaussnoise} achieves
error nearly equal to $\opt_{\eps, \delta}(A)$. Among error bounds
that depend only on $A$ and not on the database $x$, this is the best
we can hope for. Also, inspecting the proof of
Theorem~\ref{thm:dense}, one notices that
Algorithm~\ref{alg:gaussnoise} is optimal when $n =
\Omega(d/\eps)$. However, since the pioneering work of Blum, Ligett,
and Roth~\cite{BLR}, it has been known that there exist algorithms
that achieve smaller error when the input database has small size,
i.e.~$n = o(d/\eps)$. } In this subsection we present an algorithm with
stronger accuracy guarantees than Algorithm~\ref{alg:gaussnoise}: it
is optimal for any query matrix $A$ and any database size bound
$n$ (Theorem~\ref{thm:epsdelsparse}). The algorithm combines the noise distribution of
Algorithm~\ref{alg:gaussnoise} with a least squares estimation
step. Privacy is guaranteed by noise addition, while the least squares
estimation step reduces the error significantly when $n =
o(d/\eps)$. The algorithm is shown as Algorithm~\ref{alg:lse}.

\begin{algorithm}
  \caption{\textsc{Least Squares Mechanism}}\label{alg:lse}
  \begin{algorithmic}
    \REQUIRE \emph{(Public)}: query matrix $A = (a_i)_{i = 1}^N \in \R^{d
      \times N}$ ($\rank A = d$); database size bound $n$
    \REQUIRE \emph{(Private)}: database $x\in\R^N$

    \STATE Let $U_1, \ldots, U_k$ be base decomposition computed by
    Algorithm~\ref{alg:base} on input $A$, where $U_i$ is an
    orthonormal basis for a space of dimension $d_i$;
    
    \STATE Let $t$ be the largest integer such that $d_t \geq \eps n$;

    \junk{$\sum_{i =  t+1}^k{d_i} \geq \eps n$;}
    
    \STATE Let $X = \sum_{i=1}^{t}{U_i}$ and $Y = \sum_{i = t+1}^k{U_i}$;
    
    \STATE Call Algorithm~\ref{alg:gaussnoise} to compute $\tilde{y} =
    \mech_g(A, x)$; 
  
    \STATE Let $\tilde{y}_1 = XX^T\tilde{y}$ and $\tilde{y}_2 =
    YY^T\tilde{y}$; 
    
    \STATE Let $\hat{y}_1 = \arg \min\{\|\tilde{y}_1 - \hat{y}_1\|_2^2:
    \hat{y}_1 \in nXX^TK\}$, where $K = AB_1$. 
    
    \STATE \textbf{Output}  $\hat{y}_1 + \tilde{y}_2$. 
  \end{algorithmic}

\end{algorithm}

\begin{theorem} \label{thm:sparse} Let $\mech_\ell(A, x, n)$ be the
  output of Algorithm~\ref{alg:lse} on input a $d \times N$ query
  matrix $A$, database size bound $n$ and private input $x$.
  $\mech_\ell(A, x)$ is $(\eps, \delta)$-differentially private and
  for all small enough $\eps$ and all $\delta$ small enough with
  respect to $\eps$ \stocoption{has error}{satisfies}
\stocoption{\begin{equation*}O(\log^{3/2} d \sqrt{\log N \log (1/\delta)} +
    \log^2 d \log (1/\delta)) \opt_{\eps, \delta}(A, n). 
  \end{equation*}}{   
  \begin{equation*}
    \err_{\mech_\ell}(A) = O(\log^{3/2} d \sqrt{\log N \log (1/\delta)} +
    \log^2 d \log (1/\delta)) \opt_{\eps, \delta}(A, n). 
  \end{equation*} } 
\end{theorem}

\begin{lemma}\label{lm:sparse-priv}
  $\mech_\ell(A,x, n)$ satisfies $(\eps, \delta)$-differential privacy. 
\end{lemma}
\begin{proof}
  The output of $\mech_\ell(A, x, n)$ is a deterministic function of
  the output of $\mech_g(A, x)$. By Lemma~\ref{lm:dense-priv},
  $\mech_g(A, x)$ satisfies $(\eps, \delta)$-differential privacy,
  and, therefore, by Lemma~\ref{lm:simple-composition} (i.e. the
  post-processing property of differential privacy), $\mech_\ell(A, x,
  n)$ satisfies $(\eps, \delta)$-differential privacy.
\end{proof}

\begin{lemma}\label{lm:sparse-util}
  Let $U_i$ and $t$ be as defined in Algorithm~\ref{alg:lse} and let
  $r_i$ and $w_i$ be as defined in Algorithm~\ref{alg:gaussnoise}. Then $\E \max_{j = 1}^N{n |\langle a_j, r_i U_i w_i\rangle}| \leq
  O(\sqrt{\log N \log (1/\delta)})\opt_{\eps, \delta}(A, n)$. 
\end{lemma}
\begin{proof}
  \junk{
    In order to simplify notation, let us assume once again that $d$ is
    a power of $2$. Let us also use the notation for $E$, $F$,
    $\sigma_i$, $U_1$, and $V$ from
    Algorithm~\ref{alg:base}. 
    
    Analogously to the proof of Lemma~\ref{lm:dense-util}, by applying
    an appropriate unitary transformation to the columns of $A$, we may
    assume that the major axes of $E$ are aligned with the standard
    basis vectors $e_1, \ldots, e_d$, and, therefore, $F$ can be taken
    to be a diagonal matrix with $F_{ii} = \sigma_i$ the $i$-th singular
    value of $F$. Then $U_1$ is diagonal diagonal and equal to the
    projection onto $e_{d/2+1}, \ldots, e_d$ and $V$ is also diagonal
    and equal to the projection onto $e_1, \ldots, e_{d/2}$. We show
    that, as long as $t \geq 1$, 
    \begin{equation}
      \label{eq:U1-sparse-ub}
      \E \max_{j = 1}^N{n |\langle a_j, r_1 U_1 w_i\rangle|} \leq
      O(\log N \log \frac{1}{\delta})\opt_{\eps, \delta}(A, n) 
    \end{equation}
    Applying (\ref{eq:U1-sparse-ub}) recursively to $V^TA$ implies the
    statement of the lemma for all $i \leq t$, since $\opt_{\eps,
      \delta}(V^TA, n) \leq \opt_{\eps, \delta}(A, n)$. } 

  First we upper bound $\E \max_{j = 1}^N{n |\langle a_j, r_i U_i
    w_i\rangle|}$. After rearranging the terms, we have $\langle a_j,
  r_i U_i w_i\rangle = \langle U_i^Ta_j, r_i w_i \rangle$ for any $i$
  and $j$. By the definition of $r_i$, $\|U_i^Ta_j\|_2 \leq
  r_i$. Therefore, $\langle U_i^Ta_j, r_i w_i \rangle$ is a Gaussian
  random variable with mean $0$ and standard deviation at most $r_i^2c(\eps,
  \delta)$. Using standard bounds on the expected maximum of $N$
  Gaussians (e.g.~using a Chernoff bound and a union bound), we have
  that
  \begin{align*}
    \E \max_{j = 1}^N{n |\langle a_j, r_i U_i  w_i\rangle|} &= n \E
    \max_{j = 1}^N |\langle U_i^Ta_j, r_i w_i \rangle|\\ 
    &= O(\sqrt{\log N} c(\eps, \delta) nr_i^2)\\
    &=O(\sqrt{\log N \log (1/\delta)} \frac{1}{\eps} nr_i^2).  
  \end{align*}
  To finish the proof of the lemma, we need to lower bound
  $\opt_{\eps, \delta}(A, n)$ by $\Omega(\frac{1}{\eps} nr_i^2)$. We
  will use Lemma~\ref{lm:base} to lower bound $\specLB(A,
  \eps n)$ by $\Omega(\eps nr_i^2)$ and then we will
  invoke Lemma~\ref{lm:eps} to get the right dependence on $\eps$.

  By Lemma~\ref{lm:base}, for every $i$ there exists a projection
  matrix $\Pi_i = W_iW_i^T$ and a set $S_i$ such that
  $\sigma_{\min}^2(\Pi_iA|_{S_i}) = \Omega(r_i^2)$, and, furthermore,
  $|S_i| = \Omega(d_i)$. By the definition of $t$, for all $i \leq t$,
  $d_i \geq \eps n$, and, therefore, $|S_i| =
  \Omega(\eps n)$. Take $T_i \subseteq S_i$ to be an arbitrary
  subset of $S_i$ of size $\Omega(\eps n)$. The smallest
  singular value of $\Pi_iA|_{S_i}$ is a lower bound on the smallest
  singular value of $\Pi_iA|_{T_i}$:
  \begin{align*}
    \sigma_{\min}(\Pi_iA|_{T_i}) &= \stocoption{}{\min_{x: \|x\|_2=
      1}{\|(\Pi_iA|_{T_i})x\|_2}=}
     \min_{\substack{x: \|x\|_2=   1\\ \supp(x) \subseteq T_i}}{\|(\Pi_iA|_{S_i})x\|_2}\\
    &\geq \min_{x: \|x\|_2=   1}{\|(\Pi_iA|_{S_i})x\|_2} = \sigma_{\min}(\Pi_iA|_{S_i}),
  \end{align*}
  where $\supp(x)$ is the subset of coordinates on which $x$ is
  nonzero. Therefore, $\sigma_{\min}^2(\Pi_iA|_{T_i}) =
  \Omega(r_i^2)$. Substituting into the definition of $\specLB(A,
  \eps n)$, we have 
  \begin{equation*}
    \specLB(A, \eps n) = \Omega(|T_i| \sigma_{\min}(\Pi_iA|_{T_i})) =
    \Omega(\eps nr_i^2).  
  \end{equation*}
  Therefore, by (\ref{eq:speclb}), $\opt_{c_1, c_2}(A, \eps n)
  = \Omega(\eps nr_i^2)$ for all $i \leq t$. Finally, by
  Lemma~\ref{lm:eps}, there exists a small enough $\delta =
  \delta(\eps)$, for which $\opt_{\eps, \delta}(A, n) =
  \Omega(\frac{n}{\eps}r_i^2)$, and this completes the proof.
\end{proof}
\stocoption{Using Lemmas~\ref{lm:sparse-priv} and~\ref{lm:sparse-util}, along with ideas from the dense case, we complete the proof of Theorem~\ref{thm:sparse} in the appendix}{
\begin{proofof}{Theorem~\ref{thm:sparse}}
  The privacy guarantee is direct from
  Lemma~\ref{lm:sparse-priv}. Next we prove that the error of
  $\mech_\ell$ is near optimal. 

  We will bound $\err_{\mech_\ell}(A, n)$ by bounding $\E \|\hat{y}_1
  + \tilde{y}_2 - Ax\|_2^2$ for any $x$. Let us fix $x$ and define $y
  = Ax$; furthermore, define $y_1 = XX^Ty$ and $y_2 = YY^Ty$. By the
  Pythagorean theorem, we can write 
  \begin{equation}
    \label{eq:sparse-err-terms}
    \E \|\hat{y}_1 + \tilde{y}_2 -  y\|_2^2 = \E \|\hat{y}_1 -
    y_1\|_2^2 + \E \|\tilde{y}_2 -  y_2\|_2^2.   
  \end{equation}
  We show that the first term on the right hand side of
  (\ref{eq:sparse-err-terms}) is at most $O(\log^{3/2}d \sqrt{\log N
    \log(1/\delta)})\opt_{\eps,\delta}(A, n)$ and the second term is
  $O(\log^2d \log(1/\delta))\opt_{\eps,  \delta}(A, n)$ .  

  The bound on $\E \|\tilde{y}_2 -  y_2\|_2^2$ follows from
  Theorem~\ref{thm:dense}. More precisely, $Y^T\tilde{y_2}$ is
  distributed identically to the output of $\mech_g(Y^TA, x)$, and by
  Theorem~\ref{thm:dense}, 
 \stocoption{\begin{align*}
     \E \|\tilde{y}_2 -  y_2\|_2^2 &=  \E \|Y^T\tilde{y}_2 -
     Y^TAx\|_2^2 \\ &=O(\log^2d\sqrt{\log N \log(1/\delta)})\opt_{\eps,
       \delta}(A, \frac{d_{t+1}}{\eps}). 
  \end{align*}}{
  \begin{equation*}
     \E \|\tilde{y}_2 -  y_2\|_2^2 =  \E \|Y^T\tilde{y}_2 -
     Y^TAx\|_2^2 =O(\log^2d\sqrt{\log N \log(1/\delta)})\opt_{\eps,
       \delta}(A, \frac{d_{t+1}}{\eps}). 
  \end{equation*}}
  Since, by the definition of $t$, $\frac{d_{t+1}}{\eps} < n$, we have
  the desired bound. 

  The bound on $\E \|\hat{y}_1 - y_1\|_2^2$ follows from
  Lemma~\ref{lm:lse} and Lemma~\ref{lm:sparse-util}. We will use the
  notations for  $w_i$, and $r_i$ defined in
  Algorithm~\ref{alg:gaussnoise}. Let $L = nXX^TK$; by
  Lemma~\ref{lm:lse},
  \begin{equation*}
    \E \|\hat{y}_1 -  y_1\|_2^2 \leq 4 \E \|\hat{y}_1 -
    y_1\|_{L^\circ} = 4\E \max_{j = 1}^N |\langle na_j, XX^T(\tilde{y}
    - y)\rangle|. 
  \end{equation*}
  The last equality follows from the definition of the dual norm
  $\|\cdot\|_{L^\circ}$ and from the fact that $L$ is a polytope with
  vertices $\{a_j\}_{j = 1}^N$, so any linear functional on $L$ is
  maximized at one of the vertices. From the fact that for all $i \neq
  j$ we have $U_i^T U_j = 0$, from the triangle inequality, and from
  Lemma~\ref{lm:sparse-util}, we derive
\stocoption{\begin{align*}
    \E \max_{j = 1}^N |\langle na_j,  XX^T&(\tilde{y} - y)\rangle| = \E \max_{j =
      1}^N n |\sqrt{k}\sum_{i = 1}^t{\langle a_j, r_iU_iw_i\rangle}|\\
    &\leq \sqrt{k}\sum_{i = 1}^t{\E \max_{j = 1}^N{n |\langle a_j, r_iU_iw_i\rangle}|}\\
    &\leq O(\sqrt{k} t \sqrt{\log N \log (1/\delta)}) \opt_{\eps, \delta}(A, n)\\
    &= O(\log^{3/2} d  \sqrt{\log N \log (1/\delta)}) \opt_{\eps,
      \delta}(A, n). 
  \end{align*}}{
  \begin{align*}
    \E \max_{j = 1}^N |\langle na_j,  XX^T(\tilde{y} - y)\rangle| &= \E \max_{j =
      1}^N n |\sqrt{k}\sum_{i = 1}^t{\langle a_j, r_iU_iw_i\rangle}|\\
    &\leq \sqrt{k}\sum_{i = 1}^t{\E \max_{j = 1}^N{n |\langle a_j, r_iU_iw_i\rangle}|}\\
    &\leq O(\sqrt{k} t \sqrt{\log N \log (1/\delta)}) \opt_{\eps, \delta}(A, n)\\
    &= O(\log^{3/2} d  \sqrt{\log N \log (1/\delta)}) \opt_{\eps,
      \delta}(A, n). 
  \end{align*}}
  This completes the proof. 
\end{proofof}}

\stocoption{We remark that the mechanisms presented here can be easily implemented in time $Nd^{O(1)}$. The computation of an approximate MEE and approximate least squares can both be done using the Frank-Wolfe algorithm~\cite{frank-wolfe}. A nice feature of our mechanisms is that using heuristics that do not give any guarantees only affects the accuracy; the privacy holds as long as the ellipsoid used is a containing ellipsoid, and we sample noise accurately from the Gaussian. We discuss these issues in more detail in the full version.}{}
\ifstoc
\else
\subsection{Computational Complexity}

In this subsection we consider the computational complexity of our
algorithms. We pay special attention to approximating the minimum
enclosing ellipsoid of a polytope and computing least squares
estimators. For both problems we need to go into the properties of
known approximation algorithms in order to verify that the
approximations are sufficient to guarantee that our algorithms can be
implemented in polynomial time without hurting their
near-optimality. 

\subsubsection{Minimum Enclosing Ellipsoid}

Computationally the most expensive step of the base decomposition
algorithm (Algorithm~\ref{alg:base}) is computing the minimum
enclosing ellipsoid $E$ of $K$. Computing the exact MEE can be costly:
the fastest known algorithms have complexity on the order of
$d^{O(d)}N$~\cite{adler1993randomized}. However, for our purposes it
is enough to compute an approximation of $E$ in Banach-Mazur distance,
i.e. some ellipsoid $E'$ satisfying $\frac{1}{C}E' \subseteq E
\subseteq E'$ for an absolute constant $C$. Known approximation
algorithms for MME guarantee that their output is an enclosing
ellipsoid with volume approximately equal to that of the
MEE~\cite{khachMEE,modKhaMEE}. It is not immediately clear whether
such an approximation is also a Banach-Mazur approximation. However,
we can use the fact that the algorithms
in~~\cite{khachMEE,kumar2005minimum,modKhaMEE} output an ellipsoid
$E'$ satisfying approximate complimentary slackness conditions and
show that $\Pi E'$ approximates $\Pi E$ in Banach-Mazur sense for some
projection $\Pi$ onto a subspace of dimension $\Omega(d)$. This
suffices for a slightly weaker version of Lemma~\ref{lm:base}.

We begin with a definition. Let's define a vector $p \in [0,1]^N$ to
be \emph{$C$-optimal} for $A = (a_i)_{i=1}^N$ if the following conditions are
satisfied:
\begin{itemize}
\item $\sum_{i=1}^N{p_i} = 1$;
\item for all $i \in [N]$, $a_i^T (APA^T)^{-1} a_i \leq C \cdot d$
  where $P = \diag(p)$ (we use this notation throughout this
  section). 
\end{itemize}
The $C$-optimality conditions are a relaxation of the
Karush-Kuhn-Tucker conditions of a formulation of the MEE problem as
convex program. The approximation algorithm for MEE due to
Khachiyan~\cite{khachMEE}, and later follow up
work~\cite{kumar2005minimum,modKhaMEE} compute a $C$-optimal $p$ and
output an approximate MEE ellipsoid $E(p) = \{x: x^T (APA^T)^{-1} x
\leq Cd\}$\junk{, or, equivalently, $E(p) = Cd AP B_2^N$}. The running
time of the algorithm in~\cite{modKhaMEE} is $\tilde{O}(d^2N)$, where
the $\tilde{O}$ notation suppresses dependence on $C$ as well as
polylogarithmic terms.

$C$-optimality implies the following property of the the ellipsoid $E(p)$
which is key to our analysis. 
\begin{lemma}\label{lm:apx-MEE}
  Let $E^* = F^*B_2^d$ be the minimum enclosing ellipsoid of $K =
  \sym\{a_i\}_{i=1}^N$, and let $p$ be $C$-optimal for $A =
  (a_i)_{i=1}^N$. Let also $E(p) = \{x: x^T (APA^T)^{-1} x \leq Cd\} =
   F(p)B_2^d$. Then,
  \begin{equation*}
    \sigma_{d/2}^2(F(p)) \leq 4C\sigma_{d/4}^2(F^*).
  \end{equation*}
\end{lemma}
\begin{proof}
  Let $G = (F^*)^{-1}$. Since $GE^* = B_2^d$, we have that the MEE of
  $GK$ is the unit ball, and, therefore, $\|Ga_i\|_2^2 \leq 1$ for all
  $i \in [N]$. 
  Since $F(p)F(p)^T = Cd\cdot APA^T$, we have
  \begin{align*}
    \frac{1}{d}\sum_{i =1 }^N{\sigma_i^2(GF(p))} &=  \frac{1}{d}
    \tr(GF(p)F(p)^TG^T) \\
    &= C\tr(GAPA^TG^T)  \\
    &= C \sum_{i = 1}^N{p_i\|Ga_i\|_2^2} \leq C. 
  \end{align*}
  By Markov's inequality, $\sigma_{d/4}^2(GF(p)) \leq 4C$. Let $\Pi_1$
  be the projection operator onto the subspace spanned by the left
  singular vectors of $GF(p)$ corresponding to $\sigma_{d/4}(GF(p)), \ldots,
  \sigma_{d}(GF(p))$. We have $\Pi_1GE(p) \subseteq 2C^{1/2} \Pi_1
  B_2^d$. Multiplying on both sides by $F^*$, we get
  \begin{equation*}
    F^*\Pi_1 G E(p) \subseteq 2C^{1/2} F^* \Pi_1 B_2^d.
  \end{equation*}
  Let $\Pi_2$ be the matrix $\Pi_2 = G^{-1}\Pi_1G = F^*\Pi_1G$. Since
  $\Pi_2$ is similar to $\Pi_1$, it is also a projection matrix onto a
  $3d/4$ dimensional subspace. We have that 
  $F^*\Pi_1 = \Pi_2 F^*$, and therefore
  \begin{equation*}
    \Pi_2 E(p) \subseteq 2C^{1/2} \Pi_2 E^*.
  \end{equation*}
  Define $H = 4C F^*(F^*)^T - F(p)F(p)^T$. The inclusion above is
  equivalent to the positive semidefiniteness of the matrix
  $\Pi_2H\Pi_2^T$. As $\Pi_2$ is a projection onto a $3d/4$
  dimensional subspace, by the standard minimax characterization of
  eigenvalues we have  $\lambda_{3d/4}(H) \geq 0$. We recall
  the (dual) Weyl inequalities for symmetric $d \times d$ matrices $X$ and $Y$:
  \begin{equation*}
    \lambda_i(X) + \lambda_j(Y) \leq \lambda_{i + j - d}(X + Y).
  \end{equation*}
  The inequalities are standard and follow from the minimax
  characterization of eigenvalues and dimension counting arguments ---
  see, e.g.~Chapter 1 in~\cite{tao-randmatrix}. Substituting $X = H$
  and $Y = F(p)F(p)^T$, $i = 3d/4$ and $j = d/2$, we have the
  inequality
  \begin{equation*}
    \sigma_{d/2}^2(F(p)) = \lambda_{d/2}(F(p)F(p)^T) \leq
    \lambda_{d/4}(4C F^*(F^*)^T) = 4C \sigma_{d/4}^2(F^*),
  \end{equation*}
  and the proof is complete. 
\end{proof}

Finally we give an analogue of Lemma~\ref{lm:base} for the variant of
the base decomposition algorithm that uses an approximate MEE. The
proof follows from Lemma~\ref{lm:apx-MEE} and the arguments used to
prove Lemma~\ref{lm:base}. We omit a full proof here.
\begin{lemma}\label{lm:base-efficient}
  Consider a variant of Algorithm~\ref{alg:base} that, at each step,
  uses $E(p) = \{x: x^T (APA^T)^{-1} x \leq Cd\} = F(p)B_2^d$, where
  $p$ is $O(1)$-optimal for $A$, rather than the minimum enclosing
  ellipsoid $E = FB_2^d$.  Let $d_i$ be the dimension of the
  span of $U_i$. For any $i$ there exists a subspace
  $W_i$ of dimension $\Omega(d_i)$ and a set $S_i \subseteq [N]$ of
  size $|S_i| = \Omega(d_i)$, such that
  $\sigma_{\min}^2(W_iW_i^TA|_{S_i}) = \Omega(1) \max_{j
    =1}^N{\|U_i^Ta_j\|_2^2}$.
\end{lemma}
One can verify that in all our proofs we can substitute
Lemma~\ref{lm:base-efficient} for Lemma~\ref{lm:base} without changing
the asymptotics of our lower and upper bounds. Therefore, in all our
algorithms we can use the variant of Algorithm~\ref{alg:base} from
Lemma~\ref{lm:base-efficient} \junk{that uses an approximate MEE} without
compromising near-optimality. This variant of Algorithm~\ref{alg:base}
runs in time $d^{O(1)}N$. 

\junk{
modified version~\cite{modKhaMEE} of Khachiyan's classical
approximation algorithm for MEE~\cite{khachMEE}. The algorithm
in~\cite{modKhaMEE} runs in time $\tilde{O}(d^2N)$, where the
$\tilde{O}$ notation suppresses dependence on $\alpha$ as well as
lower-order terms.}

Notice that the base decomposition can be reused for different
databases, as long as the query matrix $A$ stays unchanged; once the
decomposition is computed the rest of the algorithm is very efficient:
it involves some standard algebraic computations and sampling from an
$O(d)$-dimensional gaussian distribution. Furthermore, any ellipsoid
$E'$ containing $K$ suffices for privacy, and one may use heuristic
approximations to the MEE problem.

\subsubsection{Least Squares Estimator}

Except for base decomposition, the other potentially computationally
expensive step in Algorithm~\ref{alg:lse} is the computation of a
least squares estimator $\hat{y}_1$. This is a quadratic minimization
problem, and can be approximated by the simple Frank-Wolfe gradient
descent algorithm~\cite{frank-wolfe}. In particular, for a point
$\hat{y}'$ such that $\|\hat{y}' - y\|_2^2 \leq \min_{\hat{y} \in
  L}{\|\hat{y}' - y\|_2^2} + \alpha$, Lemma~\ref{lm:lse} holds to
within an additive approximation factor $\alpha$, i.e. $\|\hat{y}' -
y\|_2^2 \leq 4\|w\|_{L^\circ} + \alpha$. We call such a point
$\hat{y}'$ an $\alpha$ additive approximation to the least squares
estimator problem. By the analysis of Clarkson~\cite{clarkson}, $T$
iterations of the Frank-Wolfe algorithm give an additive approximation
 where $\alpha \leq 4C(L)/(T+3)$, for 
  $C(L) \leq \sup_{u, v \in L}{|\langle u, u-v \rangle|}$. 
In our case $L = nXX^TK$. In order to have near
optimality for Algorithm~\ref{alg:lse}, an additive approximation
$\alpha \leq \sum_{i = 1}^t{nr_i^2}$ suffices. Using the triangle
inequality and Cauchy-Schwarz, we can bound $C(L)$ for $L = nXX^TK$ as
\begin{equation*}
  C(L) \leq \sum_{i = 1}^t{\sup_{u, v \in nU_iU_i^TK}{\langle u^T, u-v
    \rangle}} \leq \sum_{i=1}^t{2n^2r_i^2}. 
\end{equation*}
Therefore, $T = O(n)$ iterations of the Frank-Wolfe algorithm
suffice. Since each iteration of the algorithm involves $N$ dot
product computations and solving a homogeneous linear system in at
most $d$ variables and at most $d$ equations, it follows that an
approximate version of Algorithm~\ref{alg:lse} with unchanged
asymptotic optimality guarantees can be implemented in time
$d^{O(1)}Nn$.

We note that the approximation algorithm of Khachiyan for the MEE
problem, as well as its modification in~\cite{modKhaMEE}, can also be
interpreted as instances of the Frank-Wolfe algorithm
(see~\cite{modKhaMEE} for details). 

\junk{The lower bounds matching the error guarantees of
Algorithms~\ref{alg:gaussnoise} and~\ref{alg:lse} are also efficiently
computable. In addition to efficiently approximating the MEE of a
convex polytope, to compute the lower bound value we also need an
efficient algorithm that can find the subset of contact points
guaranteed to exist by Theorem~\ref{thm:bt}. One such algorithm is
given by the constructive proof  by Spielman and
Srivastava~\cite{bt-constructive} of the restricted invertability
theorem of Vershynin (which, together with John's theorem, easily
implies Theorem~\ref{thm:bt}). }

\section{Results for Pure Privacy}\label{sect:pure}

Our geometric approach to approximate privacy allows us to better
understand the optimal error required for approximate privacy vs.~that
required for pure privacy. Our first result bounds the gap between the
optimal error bounds for the two notions of privacy in the dense
case. Then we extend these ideas and give a $(\eps, 0)$-differentially
private algorithm which nearly matches the guarantees of
Algorithm~\ref{alg:lse} for sparse databases.

\subsection{The Cost of Pure Privacy}\label{sect:cost}

In this subsection we investigate the worst-case gap between $\opt_{\eps,
  \delta}(A)$ (for small enough $\delta > 0$) and $\opt_{\eps, 0}(A)$
over all query matrices $A$. At the core of our analysis is a natural geometric
fact: for any symmetric polytope $K$ with $N$ vertices in $d$-dimensional
space we can find a subset of $d$ vertices of $K$ whose
symmetric convex hull has volume radius at most a factor $O(\sqrt{\log
  (N/d)})$ smaller than the volume radius of $K$. Our proof of this
fact goes through analyzing the contact points of $K$ with its minimum
enclosing volume ellipsoid, and  a bound on the volume of
polytopes with few vertices.

\begin{theorem}\label{thm:d-vert}
  Let $K = \sym\{a_1, \ldots, a_N\} \subseteq \R^d$ and let $E$ be an
  ellipsoid of minimal volume containing $K$. There exists a set $S
  \subseteq [N]$ of size $d$ such that the matrix $A|_S = (a_i)_{i \in S}$
  satisfies $\det(A|_S)^{1/d} = \Omega(\vrad(E))$.
\end{theorem}

For the proof of~\ref{thm:d-vert} we will use John's theorem
(Theorem~\ref{thm:john}) and the following elementary algebraic result.
\begin{lemma}\label{lm:det-avg}
  Let $u_1, \ldots, u_m$ be $d$-dimensional unit vectors and let $c_1,
  \ldots, c_m$ be positive reals such that
  \begin{equation} \label{eq:id-decomp}
    \sum{c_i u_i u_i^T} = I.
  \end{equation}
  Then there exists a set $S \subseteq [m]$ of size $d$ such that the
  matrix $U = (u_i)_{i \in S}$ satisfies $|\det(U)|^{1/d} = \Omega(1)$. 
\end{lemma}
\begin{proof}
  Notice that $\tr(u_iu_i^T) = \|u_i\|_2^2 = 1$ for all $i$. By
  taking traces of both sides of (\ref{eq:id-decomp}), we have
  $\sum{c_i} = d$.

  Let $U = (u_i)_{i = 1}^m$ and let $C$ be the $m \times m$ diagonal
  matrix with $(c_i)_{i=1}^m$ on the diagonal. Then we can write $I =
  \sum{c_i u_i u_i^T} = UCU^T$. By the Binet-Cauchy formula for
  the determinant, 
  \begin{align}
    1 = \det(UCU) &= \sum_{S \subseteq [m]: |S| = d}{\prod_{i \in
        S}{c_i}\det(U|_S)^2}\notag\\
    &\leq \max_{S \subseteq [m]: |S| = d}{\det(U|_S)^2} \sum_{S
      \subseteq [m]: |S| = d}{\prod_{i \in S}{c_i}}\notag\\
    &\leq \max_{S \subseteq [m]: |S| = d}{\det(U|_S)^2} \frac{1}{d!}
    \left(\sum_{i = 1}^m{c_i}\right)^d \label{eq:prod-sum}\\
    &\leq \max_{S \subseteq [m]: |S| = d}{\det(U|_S)^2} \frac{d^d}{d!}\notag
  \end{align}
  The inequality (\ref{eq:prod-sum}) follows since each term $\prod_{i
    \in S}{c_i}$ appears $d!$ times in the expansion of
  $(\sum{c_i})^d$ and all other terms in the expansion are
  positive. Using the inequality $d! \geq (d/e)^d$, we have that
  $\max_{S \subseteq [m]: |S| = d}{\det(U|_S)^{2/d}} \geq 1/e$, and
  this completes the proof.
\end{proof}

\begin{proofof}{Theorem~\ref{thm:d-vert}}
  We can write the minimum enclosing ellipsoid $E$ as $\vrad(E)FB_2$
  where $F$ is a linear map with determinant $1$. Since $F^{-1}$ does
  not change volumes, $B_2$ is a minimal volume ellipsoid of
  $\vrad(E)^{-1}F^{-1}K$. Also, for any $A|_S = (a_i)_{i \in S}$,
  where $S \subseteq [N]$, we have 
  $$\det(A|_S) =
  \vrad(E)^d\det(\vrad(E)^{-1}F^{-1}A|_S).$$
  Therefore, it is
  sufficient to show that for $L = \sym\{u_1, \ldots, u_N\}$ such
  that $B_2$ is the minimal volume ellipsoid of $L$, there exists a
  set $S \subseteq [N]$ such that the matrix $U|_S = (u_i)_{i \in S}$
  satisfies $\det(U|_S)^{1/d} = \Omega(1)$. Since, by convexity, the
  contact points $L \cap B_2$ of $L$ are a subset of $u_1, \ldots,
  u_N$, the statement follows from Theorem~\ref{thm:john} and
  Lemma~\ref{lm:det-avg}.
\end{proofof}

Combined with the following theorem of B\'{a}r\'{a}ny and
F\"{u}redi~\cite{barany}, and also Gluskin~\cite{gluskin} (with
sharper bounds), Theorem~\ref{thm:d-vert} implies the corollary that for
any $d$-dimensional symmetric polytope one can find a set of $d$
vertices whose symmetric convex hull captures a significant fraction
of the volume of the polytope.

\stocoption{\begin{theorem}[\cite{barany,gluskin}]~\label{thm:polyt-vol-apx} 
  Let $K = \sym\{a_1, \ldots, a_N\}$ and let $E$ be a containing ellipsoid. Then $\vrad(K) \leq O\left(\sqrt{\frac{\log
        (N/d)}{d}}\right)\vrad(E)$.
\end{theorem}}{
\begin{theorem}[\cite{barany,gluskin}]~\label{thm:polyt-vol-apx} 
  Let $K = \sym\{a_1, \ldots, a_N\}$ and let $E$ be an ellipsoid
  containing $K$. Then $\vrad(K) \leq O\left(\sqrt{\frac{\log
        (N/d)}{d}}\right)\vrad(E)$.
\end{theorem}}

\begin{cor}\label{cor:d-vert-vrad}
  For any  $K = \sym\{a_1, \ldots, a_N\}$ there exists a set $S
  \subseteq [N]$ such that 
  $$
  \det((a_i)_{i \in S})^{1/d} =  \Omega\left(\sqrt{\frac{d}{\log
        (N/d)}}\right) \vrad(K).
  $$
\end{cor}

\junk{
\begin{cor}
  For any  $K = \sym\{a_1, \ldots, a_N\}$ there exists a set $S
  \subseteq [N]$ such that 
  $$
  \vrad(\sym\{a_i\}_{i \in S}) =  \Omega\left(\frac{1}{\sqrt{\log
        (N/d)}}\right) \vrad(K). 
  $$
\end{cor}
\begin{proof}
  Let $E$ be the minimal volume ellipsoid of $K$. By
  Theorem~\ref{thm:d-vert}, there exists a set $S\subseteq [N]$ of size
  $d$ such that $\det(A|_S)^{1/d} = \Omega(\vrad(E))$ (where $A|_S =
  (a_i)_{i \in S}$). Therefore,
  \begin{align}
    \vrad((\sym\{a_i\}_{i \in S})) &=
    \frac{\vol((A|_S)B_1)^{1/d}}{\vol(B_2^d)^{1/d}}\\
    &= \frac{\det(A|_S)^{1/d}\vol(B_1)^{1/d}}{\vol(B_2^d)^{1/d}}\\
    &= \Omega\left(\frac{1}{\sqrt{d}}\vrad(E)\right). 
  \end{align}
  Finally, by Theorem~\ref{thm:polyt-vol-apx}, $\vrad(E) =
  \Omega(\sqrt{d/\log N}) \vrad(K)$.
\end{proof}
}

Finally, we describe the application to differential privacy. 
By Corollary~\ref{cor:d-vert-vrad},  $\vollb(A, \eps) =
O(\frac{1}{\eps^2}\log (N/d)) \detlb(A, d)$. Also, by
(\ref{eq:detlb}), $\detlb(A, d) = O(\log^2 d)\opt_{c_1, c_2}(A,
d)$. Finally, Lemma~\ref{lm:eps} implies that
$\opt_{c_1, c_2}(A, d) \leq {\eps^2} \opt_{\eps, \delta}(A,
d/\eps)$ for  $\delta$ small enough with respect to $\eps$. Putting
all this together and using Theorem\ref{thm:htvollb}, we have the following theorem (Theorem~\ref{thm:epsdelgap}). 

\begin{theorem} \label{thm:cost-pure}
  For small enough  $\eps$ and all $\delta$ small enough with respect
  to $\eps$, for any $d \times
  N$ real matrix $A$ we have
  \begin{equation*}
    \opt_{\eps, 0}(A) = O(\log^3 d)\vollb(A,\eps) =  O(\log^5 d \log
    (N/d)) \opt_{\eps, \delta}(A, \frac{d}{\eps}).  
  \end{equation*}
\end{theorem}

\subsection{Sparse Case under Pure Privacy}

We further extend our results from Section~\ref{sect:sparse} and show
an efficient $(\eps, 0)$-differentially private algorithm which, on
input any query matrix $A$ and any database size bound $n$, nearly
matches $\opt_{\eps, 0}(A, n)$. This proves our main Theorem~\ref{thm:epssparse}. In fact, our result is stronger: we
show an $(\eps, 0)$-differentially private mechanism whose error
nearly matches $\opt_{\eps, \delta}(A, n)$ for all $\delta$ small
enough with respect to $\eps$. Thus, the result of this subsection can
be seen as a generalization of Theorem~\ref{thm:cost-pure} to the
sparse databases regime.

Our algorithm for sparse databases under pure privacy closely follows
Algorithm~\ref{alg:lse}: we add noise from a distribution that is
tailored to $A$ but oblivious to the database $x$; then we use least
squares estimation to reduce error on sparse databases. However,
Gaussian noise does not preserve $(\eps, 0)$-differential privacy, and
we need to use a different noise distribution. Intuitively, one
expects that adding noise sampled from a near-optimal
distribution~\cite{HardtT10,12vollb} and then computing a least squares
estimator would be nearly optimal, analogously to
Algorithm~\ref{alg:lse}. We are not currently able to analyze the
error of this algorithm, but instead we analyze a variant of
Algorithm~\ref{alg:lse} where the Gaussian distribution is simply
substituted with the generalized $K$-norm distribution
from~\cite{HardtT10}. Intuitively, we are able to show that the
generalized $K$-norm distribution ``approximates a Gaussian'' well
enough for our analysis to go through. A main tool in our analysis is
a classical concentration of measure inequality from convex geometry. 

We begin with a slight generalization of the main upper bound result of Hardt
and Talwar~\cite{HardtT10}. This generalization follows directly from
the methods used in~\cite{HardtT10} with only minor modifications in
the proofs. We omit a full derivation here. Also, while the methods of
Hardt and Talwar will lead to a proof conditional on the truth of the
Hyperplane conjecture from convex geometry, using the ideas of
Bhaskara et al.~\cite{12vollb} the result can be made unconditional. 

\begin{theorem}[\cite{HardtT10,12vollb}]\label{thm:knorm}
  Let $A= (a_i)_{i = 1}^N$ be an $d\times N$ real matrix and let $K =
  \sym\{a_i\}_{i = 1}^N$. There exists an efficiently computable and
  efficiently sampleable distribution $\mathcal{W}(A, \eps)$ such that
  the following claims hold:
  \begin{packed_enum}
  \item the algorithm $\mech_K$ which on input $x$ outputs $Ax + w$ for $w \sim
    \mathcal{W}(A, \eps)$ satisfies $(\eps, 0)$-differential privacy;
  \item $\mathcal{W}(A, \eps)$ is identical to the distribution of the
    random variable $m \sum_{\ell = 1}^m{w_\ell}$ where $m = O(\log d)$, and
    $w_{\ell}$ is a sample from a log concave distribution with support lying
    in a subspace $\mathcal{V}_{\ell}$ of $\R^d$ of dimension $d_{\ell}$;
  \item $\mathcal{V}_{\ell}$ and $\mathcal{V_{\ell'}}$ for $\ell' \neq \ell$ are
    orthogonal, and the union of $\{\mathcal{V}_\ell\}_{\ell = 1}^m$ spans $\R^d$;
  \item let $M_\ell = M_\ell(A, \eps) = \E m^2 w_\ell w_\ell^T$ be the correlation
    matrix of $m w_\ell$ and let $\Pi_\ell$ be the projection matrix onto
    $\vspan\{\mathcal{V}_j\}_{j=\ell}^m$; then 
    $$
    \lambda_{\max}(M_\ell) \leq O(\log^2
    d)\frac{d_\ell}{\eps^2} \vrad(\Pi_\ell K)^2,
    $$
    where    $\lambda_{\max}(M_\ell)$ is the largest eigenvalue of   $M_\ell$.  
  \end{packed_enum}
\end{theorem}

Using the distribution $\mathcal{W}(A, \eps)$, we define our near
optimal sparse-case algorithm satisfying pure differential privacy as
Algorithm~\ref{alg:lse-pure}.

\begin{algorithm}
  \caption{\textsc{Least Squares Mechanism: Pure Privacy}}\label{alg:lse-pure}
  \begin{algorithmic}
    \REQUIRE \emph{(Public)}: query matrix $A = (a_i)_{i = 1}^N \in \R^{d
      \times N}$ ($\rank A = d$); database size bound $n$
    \REQUIRE \emph{(Private)}: database $x\in\R^N$

    \STATE Let $U_1, \ldots, U_k$ be base decomposition computed by
    Algorithm~\ref{alg:base} on input $A$, where $U_i$ is an
    orthonormal basis for a space of dimension $d_i$;
    
    \STATE Let $t$ be the largest integer such that $d_t \geq \eps n$;

    \FORALL{$i\leq t$}
    \STATE Compute $\tilde{y}_i = U_i(U_i^TAx + w_i)$ where $w_i \sim
    \mathcal{W}(U_i^TA, \frac{\eps}{t+1})$. 
    \ENDFOR

    \STATE Let $\tilde{y}' = \sum_{i = 1}^t{\tilde{y}_i}$

    \STATE Let $X = \sum_{i=1}^{t}{U_i}$ and $Y = \sum_{i = t+1}^k{U_i}$;    

    \STATE Compute $\tilde{y}'' = Y(Y^TAx +  w'')$ where $w'' \sim
    \mathcal{W}(Y^TA, \frac{\eps}{t+1})$;

    \STATE Compute $\hat{y}' = \arg \min\{\|\tilde{y}' - \hat{y}'\|_2^2:
    \hat{y}_1 \in nXX^TK\}$, where $K = AB_1$. 
    
    \STATE \textbf{Output}  $\hat{y}' + \tilde{y}''$. 
  \end{algorithmic}

\end{algorithm}

\begin{theorem}\label{thm:sparse-pure}
  Let $\mech_p(A, x, n)$ be the output of Algorithm~\ref{alg:lse-pure}
  on input a $d \times N$ query matrix $A$, database size bound $n$
  and private input $x$.  $\mech_p(A, x,n)$ is $(\eps,
  0)$-differentially private and for all small enough $\eps$ and all
  $\delta$ small enough with respect to $\eps$ satisfies
  \begin{align*}
    \err_{\mech_p}(A) &= O(\log^{4} d \log^{3/2} N)\opt_{\eps, \delta}(A, n) + O(\log^5 d \log N)\opt_{\eps, 0}(A, n);\\
    \err_{\mech_p}(A) &= O(\log^{4} d \log^{3/2} N + \log^7 d \log
    N)\opt_{\eps, \delta}(A, n).
  \end{align*}      
\end{theorem}

Once again privacy follows by a straightforward argument from the
privacy of the underlying noise-adding mechanism, in this case the
generalized $K$-norm mechanism.

\begin{lemma}\label{lm:sparse-pure-priv}
  $\mech_p(A,x, n)$ satisfies $(\eps, 0)$-differential privacy. 
\end{lemma}
\begin{proof}
  $\mech_p(A,x,n)$ is a deterministic function of $\tilde{y_1},
  \ldots, \tilde{y_t}$ and $\tilde{y}''$. Each of these quantities is
  the output of an algorithm satisfying $(\frac{\eps}{t+1},
  0)$-differential privacy (by Theorem~\ref{thm:knorm}, claim
  1). Therefore, by Lemma~\ref{lm:simple-composition},
  $\mech_p(A,x,n)$ satisfies $(\eps, 0)$-differential privacy. 
\end{proof}

Next we prove the main technical lemma we need in order to show near
optimality. The analysis is very similar to that of
Lemma~\ref{lm:sparse-util}. The main technical challenge is to show
that the distribution $\mathcal{W}$ has all the properties we needed
from the Gaussian distribution: covariance with bounded operator norm
and exponential concentration. We use ideas from
Section~\ref{sect:cost} and the following variant of a classical concentration of
measure inequality, due to Borell~\cite{borell} (proved in the appendix).

\begin{theorem}
\label{thm:logconcaveborell}
 Let $\mu$ be a log-concave distribution over $\R^d$. Assume that $A$ is a symmetric convex subset of $\R^d$ such that $\mu(A) = \theta \geq \frac 2 3$. Then, for every $t > 1$ we have 
 \[
 \mu[(tA)^c] \leq 2^{-(t+1)/2}
 \]
 \end{theorem}
 
We are now ready to prove the counterpart of
Lemma~\ref{lm:sparse-util} for $\mech_p$. 

\begin{lemma}\label{lm:sparse-pure-util}
  Let $U_i$, $t$, and $w_i$ be as defined in
  Algorithm~\ref{alg:lse-pure} \junk{and let $r_i$ and $w_i$ be as
    defined in Algorithm~\ref{alg:gaussnoise}}. For 
  all small enough $\eps$ and all $\delta$ small enough with respect
  to $\eps$, $\E \max_{j = 1}^N{n |\langle a_j, \sum_{i=1}^t U_iw_i\rangle|} \leq
  O(\log^{4} d\log^{3/2} N)\opt_{\eps, \delta}(A, n)$.
\end{lemma}
\begin{proof}
  As in Algorithm~\ref{alg:lse}, we define $r_i = \max_{j =
    1}^N{\|U_i^Ta_j\|_2}$. Equivalently $r_i$ is the radius of the
  smallest $d_i$-dimensional ball which contains $U_i^TK$. In the
  proof of Lemma~\ref{lm:sparse-util} we argued that $\opt_{\eps, \delta}(A,n) =
  \Omega(\frac{n}{\eps}r_i^2)$ for all small enough $\eps$ and all
  $\delta$ small enough with respect to $\eps$. 

  By Theorem~\ref{thm:knorm}, we can write $w_i$ as $m_i\sum_{\ell =
    1}^{m_i} w_{i\ell}$ where $w_{i\ell}$ is a  sample from a log concave distribution over a subspace $\mathcal{V}_{i\ell}$ and $m_i = O(\log  d_i)$. Furthermore, all $w_{i\ell}$ for a fixed $i$ are mutually
  orthogonal. Finally, letting $\Pi_{i\ell}$ be the projection matrix
  onto $\vspan\{\mathcal{V}_{i\ell}\}_{\ell = 1}^{m_i}$ and $M_{i\ell}$ be the covariance
  matrix of $m_iw_{i\ell}$, we have 
  \begin{equation}\label{eq:eigen-up}
    \lambda_{\max}(M_{i\ell})  \leq O(\log^2 d ) \frac{d_{i\ell}
      t^2}{\epsilon^2}\vrad(\Pi_{i\ell}U_i^TK)^2.  
  \end{equation}
  Therefore, for any $a_j$, we can derive the following bound:
  \begin{align*}
    \E n^2 |\langle a_j, U_iw_i\rangle|^2 &= n^2 \E |\langle U_i^Ta_j,
    w_i\rangle|^2\\
    &\leq n^2  m_i \sum_{\ell = 1}^{m_i}{\E |\langle U_i^Ta_j,
     m_i w_{i\ell}\rangle|^2}\\
    &= n^2  m_i \sum_{\ell = 1}^{m_i}{{(U_i^Ta_j)^TM_{i\ell}(U_i^Ta_j)}} \\
    &\leq n^2 m_i \sum_{\ell = 1}^{m_i}{\|U_i^Ta_j\|_2^2
      \lambda_{\max}(M_{i\ell})}\\
    &\leq O(\log^5 d) n^2 \sum_{\ell =
      1}^{m_i}{\frac{1}{\eps^2}r_i^2 d_{i\ell}\vrad(\Pi_{i\ell}U_i^TK)^2}. 
  \end{align*}
  The first bound above follows from the Cauchy-Schwarz inequality and
  the last bound follows from (\ref{eq:eigen-up}). To bound
  $d_{i\ell} \vrad(\Pi_{i\ell}U_i^TK)^2$, recall that $U_i^TK$ is
  contained in a ball of radius $r_i$, and therefore so is
  $\Pi_{i\ell}U_i^TK$ for all $\ell$. Therefore, by
  Theorem~\ref{thm:polyt-vol-apx}, $\vrad(\Pi_{i\ell}U_i^TK)^2 = O((\log
    (N/d_{i\ell}))/d_{i\ell})r_i^2$. Substituting into the bound
  above and recalling that $\opt_{\eps, \delta}(A,n) =
  \Omega(\frac{n}{\eps}r_i^2)$, we get
  \begin{equation*}
    \E n^2 |\langle a_j, U_iw_i\rangle|^2 = O(\log^6 d
     \log (N/d)) \opt_{\eps,\delta}(A,n)^2. 
  \end{equation*}
  Thus applying Cauchy Schwarz once again, we conclude
  \begin{equation*}
\E n^2|\langle a_j, \sum_{i=1}^t U_iw_i\rangle|^2 \leq   t\cdot \sum_{i=1}^t \E n^2 |\langle a_j, U_iw_i\rangle|^2 = O(\log^8 d \log (N/d) \opt_{\eps,\delta}(A,n)).
\end{equation*}
 For any $j$,  the set $\{w:|\langle
  U_i^Ta_j, \sum_i w_i\rangle| \leq T\}$ is symmetric and convex for
  any bound $T$. Then, by Chebyshev\rq{}s inequality, and Theorem~\ref{thm:logconcaveborell} there exists a constant $C$ such
  that  for any $i,j$ and $\alpha > 2$
  \begin{equation*}
    \Pr[n|\langle U_i^Ta_j, w_{i}\rangle| >   C \alpha \log N \log^{4} d \sqrt{\log (N/d)}   \frac{n}{\eps}r_i^2] < N^{-\alpha}. 
  \end{equation*}
  Using a union bound and taking expectations completes the proof. 
\end{proof}

\begin{proofof}{ Theorem~\ref{thm:sparse-pure}}
 The privacy guarantee follows from Lemmas~\ref{lm:sparse-pure-priv}. By Lemma~\ref{lm:sparse-pure-util} analogously to the proof of Theorem~\ref{thm:sparse}, we can conclude that
\begin{equation*}
 \E [ \|\hat{y}\rq{} - XX^TAx\|_2^2] \leq 4 \E \max_{j = 1}^N{n |\langle a_j, \sum_{i=1}^t U_iw_i\rangle|} \leq
  O(\log^{4} d\log^{3/2} N)\opt_{\eps, \delta}(A, n).
\end{equation*}
Moreover, by Theorem~\ref{thm:htvollb}, 
\begin{align*}
  \E [\|\tilde{y}'' - YY^TAx\|_2^2] &\leq O(\log^3 d)
  \vollb(YY^TA,\frac{\eps}{t+1}) \\
  &= O(t^2\log^3 d) \opt_{\eps,0}(YY^TA, \frac{d_t}{\eps})\\
  &= O(\log^5 d) \opt_{\eps,0}(A, n).
\end{align*}
The last bound follows since $YY^T$ is a projection matrix, $d_t\leq
n\eps$ and $t = O(\log d)$. 
Also, by Theorem~\ref{thm:cost-pure}, $\vollb(YY^TA, \frac{\eps}{t+1}) =
O(t^2\log^2 d \log (N/d_t)) \opt_{\eps, \delta}(YY^TA,
\frac{d_t}{\eps})$, and therefore we have $\E [\|\tilde{y}'' - YY^TAx\|_2^2] =
O(\log^7 d \log (N/d)) \opt_{\eps, \delta}(A,n)$.  

\junk{By
  Corollary~\ref{cor:d-vert-vrad}, each $d_i\vrad(\Pi_i K)^2$ is at
  most $O(\log(N/d_i))$ times $detLB(A,d_i)$ which in turn is bounded
  by $O(\log^2 d) \opt_{\eps,\delta}(A,d_i)$ for small enough
  constants $\eps$ and $\delta$. Since $d_i \geq \eps n$, it follows
  that
\begin{equation*}
\E [\|\tilde{y}'' - YY^TAx\|_2^2] \leq O(\log^6 d \log N) \opt_{\eps,\delta}(A,n)
\end{equation*}}
Pythagoras theorem then implies the result.
\end{proofof}


\section{Universal bounds}\label{sect:universal}
For $d$ linear sensitivity $1$ queries, there are known universal
bounds on $\err_{\eps,\delta}(A,n)$ and $\err_{\eps,0}(A,n)$. We note
that the sensitivity $1$ bound implies that the $\ell_2$ norm of each
column is at most $\sqrt{d}$. This in turn allows us to prove an upper
bound on the spectral lower bound, so that the relative guarantee
provided by Theorems~\ref{thm:dense} and~\ref{thm:sparse} can be translated to an
absolute one. The resulting bounds can be improved by polylogarithmic
factors, by using natural simplifications of the relative-error
mechanisms and analyzing them directly. We next present these
simplifications. The average per query error bounds resulting from our
mechanisms match the best known bounds for
$(\eps,\delta)$-differential privacy, and improve on the best known bounds for pure differential privacy.

For \stocoption{\epsdeldp}{$(\eps,\delta)$-differential privacy}, the best known universal upper bound when $A \in [0,1]^{d\times N}$ for the total $\ell_2^2$ error is $O(n d \log d \sqrt{\log N} \log(1/\delta)/ \epsilon)$, given by ~\cite{GuptaRU11}. We note that when $A \in [0,1]^{d \times N}$, one can use $B(0,\sqrt{d})$ as an enclosing ellipsoid for $K$. The following simple mechanism is easily seen to be \epsdeldp.

\begin{algorithm}
  \caption{\textsc{Simple Noise + Least Squares Mechanism}}\label{alg:simplelse}
  \begin{algorithmic}
    \REQUIRE \textbf{Public Input}: query matrix $A = (a_i)_{i = 1}^N \in [0,1]^{d
      \times N}$ 
    \REQUIRE \textbf{Private Input}: database $x \in \R^N$

    \STATE Let $c(\eps, \delta) =
    \frac{1+\sqrt{2\ln(1/\delta)}}{\eps}$.

    \STATE Let $r = \max_{j = 1}^N{\|a_j\|_2}$;

    \STATE Sample $w \sim N(0, c(\eps, \delta))^{d}$;
    \STATE  Let $\tilde{y}=Ax + rw$. 
    \STATE  Let $\hat{y} = \arg \min\{\|\tilde{y} - \hat{y}\|_2^2:
    \hat{y} \in nK\}$, where $K = AB_1$. 
    \STATE \textbf{Output} $\hat{y}$
  \end{algorithmic}
\end{algorithm}

\begin{theorem}
The mechanism $\mech_s$ of Algorithm~\ref{alg:simplelse} is \stocoption{\epsdeldp }{$(\eps, \delta)$-differentially private} and satisfies
\[
\err_{\mech_s}(A) \leq O(nd \log(1/\delta)\sqrt{\log N}/\eps)
\]
\end{theorem}
\begin{proof}
The privacy of the mechanism is immediate from Lemma~\ref{lm:gauss-mech-ind}.  To analyze the error, we use Lemma~\ref{lm:lse} and the fact that $L=nK$ to bound
\[
\err_{\mech_s}(A) = \E[\|\hat{y} - Ax\|_2^2] \leq 4n\|w\|_{K^{\circ}} = 4n\E \max_{j = 1}^N |\langle a_j, w\rangle|,
\]
where $\{a_j\}_{j=1}^N$ are columns of $A$. We have used the fact that the $\|w\|_{K^\circ} = \max_{a\in K} \langle a,w\rangle$ is attained at one of the vertices of $K$. Since each $\langle a_j,w\rangle$ is a Gaussian with variance $r^4c(\eps,\delta)^2$,  $|\langle a_j, w \rangle|$ exceeds $r^2 c(\eps,\delta)\sqrt{t\log N}$ with probability at most $\frac{1}{N^t}$. Taking a union bound, we conclude that this expectation of the maximum is $O(r^2c(\eps,\delta)\sqrt{\log N})$. Recall that $r =  \max_{j = 1}^N{\|a_j\|_2} \leq \sqrt{d}$. It follows that
\[
\err_{\mech_s}(A) \leq O(nd c(\eps,\delta)\sqrt{\log N})
\]
\end{proof}

Comparing with~\cite{GuptaRU11}, our bound is better by an $O(\log d)$ factor. However, the previous bound is stronger in that it guarantees expected squared error $O_{\eps,\delta}(n\sqrt{\log N}\log d)$ for every query, while we can only bound the total $\ell_2^2$ error. 

For getting pure $\eps$-DP, we simply substitute the generalized $K$-norm distribution guaranteed by Theorem~\ref{thm:knorm} instead of the Gaussian noise.
\begin{algorithm}
  \caption{\textsc{$K$-norm Noise + Least Squares Mechanism}}\label{alg:knormlse}
  \begin{algorithmic}
    \REQUIRE \textbf{Public Input}: query matrix $A = (a_i)_{i = 1}^N \in [0,1]^{d
      \times N}$ 
    \REQUIRE \textbf{Private Input}: database $x \in \R^N$
    \STATE  Let $\tilde{y}=Ax + w$ where $w \sim \mathcal{W}(A,\eps)$. 
    \STATE  Let $\hat{y} = \arg \min\{\|\tilde{y} - \hat{y}\|_2^2:
    \hat{y} \in nK\}$, where $K = AB_1$. 
    \STATE \textbf{Output} $\hat{y}$
  \end{algorithmic}
\end{algorithm}


We first observe an upper bound on the volume radius of projections of $K$. 
\begin{lemma}\label{lem:vradbound}
Let $A \in [0,1]^{d \times N}$ and let $K= \sym\{a_1, \ldots, a_N\}$. Let $\Pi^{(k)}$ be a rank $k$ orthogonal projection that maps $\R^d$ to $\R^k$.  Then
\[
\vrad(\Pi^{(k)} K) \leq O\left(\sqrt{\frac{\log
        (N/k)}{k}}\right)\sqrt{d}
\]
\end{lemma}
\begin{proof}
Since each column of $A$ has $\ell_2$ norm at most $\sqrt{d}$, $\Pi^{(k)} K$ is contained in a ball of radius $\sqrt{d}$.  Theorem~\ref{thm:polyt-vol-apx} then immediately implies the claimed bound.
\end{proof}

Now we show that Algorithm~\ref{alg:knormlse} achieves the bound claimed in Theorem~\ref{thm:countingeps}. 

\begin{theorem}
The mechanism $\mech_{sp}$ of Algorithm~\ref{alg:knormlse} is \stocoption{$(\eps,0)$-DP}{$(\eps,
0)$-differentially private} and satisfies
\[
\err_{\mech_{sp}}(A) \leq O(nd\eps^{-1} \log^{2} d \log^{\frac{3}{2}} N)
\]
\end{theorem}
\begin{proof}
The privacy property follows from Theorem~\ref{thm:knorm} and the fact
that post-processing preserves $(\eps, 0)$-differential privacy.
To  prove the error bound, we need to upper bound 
$\E_{w} [\|w\|_{K^{\circ}}]$ for a polytope $K \subseteq \Re^d$ with $N$ vertices, where $w$ is drawn from $\mathcal{W}(A,\eps)$. It therefore suffices to bound $\E_w[\max_{i=1}^N |\langle a_i, w\rangle|]$.

By Theorem~\ref{thm:knorm}, we can write $w$ as $\sum_{\ell =   1}^m {m w_{\ell}}$ where $w_{\ell}$ is drawn from a log concave distribution and $m=O(\log d)$. For a fixed $\ell$, 
\begin{align*}
  \E[|\langle a_i, mw_{\ell}\rangle|^2] &=  a_i^T M_{\ell} a_i\\
  &\leq \|a_i\|^2 \cdot \lambda_{max}(M_{\ell}) \\
  &\leq O(d\log^2 d) \frac{1}{\eps^2} d_\ell \vrad(\Pi_{\ell} K)^2.
\end{align*}
By lemma~\ref{lem:vradbound}, this is at most \stocoption{$O(\frac{d^2}{\eps^2} \log^2 d
\log N)$}{$O\left(d\log^2 d
  \log(N/d_{\ell})d/\eps^2\right) \leq O(\frac{d^2}{\eps^2} \log^2 d
\log N)$}. Then by Cauchy-Schwarz,
\begin{eqnarray*}
  \E[|\langle a_i, \sum_{\ell=1}^m mw_{\ell}\rangle|^2] &\leq & m\cdot \sum_{\ell=1}^m\E[|\langle a_i,mw_{\ell}\rangle |^2]\\
&\leq &O(\frac{d^2}{\eps^2} \log^4 d \log N)
\end{eqnarray*}
By Theorem~\ref{thm:logconcaveborell}, $\Pr[|\langle a_i,
\sum_{\ell} w_{\ell}\rangle| \geq t \log N \cdot O(\frac{d}{\eps}\log^{2} d \sqrt{\log N}) \leq N^{-\Omega(t)}$. 
It follows that $\E_{w} [\|w\|_{K^{\circ}}]$ is bounded by $O(\frac{d}{\eps}\log^{2} d \log^{\frac{3}{2}} N)$.

It then follows that $\E[\|\hat{y}-Ax\|_2^2$ is $O(nd\eps^{-1} \log^{2} d \log^{\frac{3}{2}} N)$.
\end{proof}


\section{Extensions}\label{sect:ext}

In this section we describe a couple of extensions to our results. We
show how to translate our optimality guarantees for total squared
error in the dense case regime to worst case error over queries using
the minimax theorem. We further show that our nearly optimal efficient
mechanism in the dense case regime implies a polylogarithmic
approximation to hereditary discrepancy.

\newcommand{\ai}{\ensuremath{a^{(i)}}}
\newcommand{\ait}{\ensuremath{a^{(i)T}}}

\subsection{Expected $\ell_{\infty}$ Error}

For $i \in [d]$, let $\ai$ denote the $i$th row of $A$ and let $\mech(x)_i$ denote the $i$th coordinate of the answer $\mech(x)$. Let $\err^{\ell_\infty}_{\mech}(A) = \max_{x \in \R^N}
\E[\|Ax-\mech(x)\|_{\infty}]$ denote the worst case (over databases)
expected $\ell_{\infty}$ error of a mechanism for a query $A$. Here as
usual, the expectation is taken over the internal coin tosses of
$\mech$.  Thus we measure the expected worst case error $\E[\max_i |\langle \ai, x\rangle - \mech(x)_i|]$.
Let $\opt^{\ell_\infty}(A) = \min_{\mech: \mech\text{ is
  }\epsdeldp} \err^{\ell_\infty}_\mech(A)$ denote the minimum
$\ell_{\infty}$ error over all $(\eps, \delta)$-differentially private mechanisms.  In the dense case, our results can be extended to the $\ell_{\infty}$ error at the cost of an additional $O(\log d)$ loss in the competitive ratio.

We derive such an extension in two steps. First, we give a mechanism for which the worst case expected squared error $\max_i \E[|\langle \ai, x\rangle - \mech(x)_i|^2]$ is small. We then use this mechanism as a blackbox to derive one which has small expected $\ell_{\infty}$ error. For a mechanism $\mech$, let $E^{\mech,A}_i = \E[|\langle \ai, x\rangle - \mech(x)_i|^2]$ denote the expected squared error of the mechanism in the $i$th coordinate.

\begin{theorem}
\label{thm:worstcase-expected-noise}
Let $A \in \R^{d\times N}$, $\eps,\delta$ be given. There is an \epsdeldp\ mechanism $\mech_{we}$, and a non-negative diagonal matrix $P$ with $\tr(P^2) =1$ such that
for all $i \in [d]$,
\[
E^{\mech_{we},A}_i \leq O(\log^2 d\log 1/\delta) \specLB(PA) \leq O(\log^2 d\log 1/\delta) (\opt^{\ell_\infty}(A))^2
\]
\end{theorem}
\begin{proof}
Recall that for any $A$, the mechanism $\mech_g^A$ of Section~\ref{sec:densepsdel} is \epsdeldp\ and has total expected $\ell_2^2$ error at most $\beta \triangleq O(\log^2 \log 1/\delta)$ times the lower bound $\specLB(A)$. We will use this mechanism as a subroutine to derive $\mech_{we}$.

Let $(p_1,p_2,\ldots,p_d)$ denote a probability distribution so that $p_i \geq 0$ and $\sum_i p_i =1 $. Let 
$P$ denote a diagonal matrix with entries $\sqrt{p_i}$. It is easy to see that $\err_\mech^{\ell_2^2}(PA) = \sum_i p_i E^{\mech,A}_i$. It follows that $\opt^{\ell_2^2}_{\epsilon,\delta}(PA) \leq  (\opt^{\ell_\infty}(A))^2$: indeed the mechanism $\mech$ achieving $\opt^{\ell_\infty}(A)$ gives this error bound.

Thus the mechanism $\mech_g^{PA}$ satisfies \stocoption{\begin{align*}\err_{\mech_g^{PA}}^{\ell_2^2}(PA) &\leq \beta \specLB(PA) \\& \leq \beta \opt^{\ell_2^2}(PA).\end{align*}}{
\begin{align*}\err_{\mech_g^{PA}}^{\ell_2^2}(PA) \leq \beta \specLB(PA) \leq \beta \opt^{\ell_2^2}(PA).\end{align*}}
It follows that $\sum_i p_i E^{\mech_g^{PA},A}_i$ is $\beta \cdot (\opt^{\ell_\infty}(A))^2$. Consider a two-player zero sum game where the row-player selects a query $\ai$ from $A$, and the column player picks an \epsdeldp\ mechanism $\mech$ and must pay the row player $E^{\mech,A}_i$. The above discussion shows that for any randomized strategy (given by $P$) of the row player, the column player has a strategy that guarantees payoff at most $\beta \cdot (\opt^{\ell_\infty}(A))^2$. By the minimax theorem, this also upper bounds the value of the game. Moreover, using standard constructive versions of the minimax theorems (see e.g.~\cite{AroraHK12}), one can come up with such a distribution $\mathcal{D}$ over $O(d\log d)$ \epsdeldp\ mechanisms, and a $P$
such that for all $i \in [d]$,
\stocoption{\begin{align*}\E_{\mech\sim \mathcal{D}}[E^{\mech,A}_i]  &\leq 2\beta \specLB(PA) \\ &\leq 2\beta (\opt^{\ell_\infty}(A))^2.\end{align*}}{
\begin{equation*}\E_{\mech\sim \mathcal{D}}[E^{\mech,A}_i]  \leq 2\beta
\specLB(PA) \leq 2\beta (\opt^{\ell_\infty}(A))^2.\end{equation*}}

Recall that the mechanism $\mech_g^{PA}$ is of the form $Ax + w$ where $w$ is a noise vector whose distribution is independent of $x$. Thus the pair $(\mathcal{D},P)$ can be computed without looking at the database $x$. Therefore, the mechanism $\mech_{we}$ that samples a $\mech$ from $\mathcal{D}$ and runs it on $x$ is itself \epsdeldp. The theorem follows.
\end{proof}

Using the above result, we can now construct a mechanism that has small expected worse case error.
\begin{theorem}
Let $A \in \R^{d\times N}$, $\eps,\delta$ be given. There is an \epsdeldp\ mechanism $\mech_{ew}$, and a non-negative diagonal matrix $P$ with $\tr(P^2) =1$ such that 
\stocoption{\begin{align*}
(\err^{\ell_\infty}_{\mech_{ew}}(A))^2 &\leq O(\log^4 d \log ((\log d)/\delta)) \cdot \specLB(PA)^2 \\ &\leq O(\log^4 d \log ((\log d)/\delta)) \cdot (\opt^{\ell_\infty}(A))^2
\end{align*}}
{\[
(\err^{\ell_\infty}_{\mech_{ew}}(A))^2 \leq O(\log^4 d \log ((\log d)/\delta)) \cdot \specLB(PA) \leq O(\log^4 d \log ((\log d)/\delta)) \cdot (\opt^{\ell_\infty}(A))^2
\]}
\label{thm:expected-worstcase-noise}
\end{theorem}
\begin{proof}
Our mechanism with small expected $\ell_{\infty}$ noise simply runs $L\triangleq O(\log d)$ copies  $\mech_{we}(x)$    and returns the median answer for each coordinate. In other words, if $y^j$ denotes the outcome of the $j$th run of the mechanism $\mech_{we}$, then we set $z_i = \mbox{median}(y^1_i,\ldots,y^L_i)$.

By Markov\rq{}s inequality, we know that for each $i,j$, $\Pr[(y^j_i -
(Ax)_i)^{2} \geq kE^{\mech_{ew},A}_i]  \leq \frac{1}{k}$. Applying
Chernoff bounds, we conclude that $\Pr[(z_i - (Ax)_i)^2 \geq
kE^{\mech_{ew},A}_i] \leq (2/k)^{CL}$ for a constant $C$. By a union bound, we conclude that $\Pr[\|z-Ax\|_{\infty} \geq O(\log d)  \opt^{\ell_\infty}(A)] \leq d\cdot(\frac{2}{k})^{CL}$. Taking $L = \Omega(\log d)$ suffices to ensure that the integral of this probability over $k \in [4, \infty)$ converges, giving a bound on the expected $\ell_{\infty}$ norm of the error.

The resulting mechanism is not necessarily  $(\eps,
\delta)$-differentially private any more. However, since its
outcome can be computed by postprocessing the outcome of $L$  $(\eps,
\delta)$-differentially private mechanisms, it is still
$(L\eps,L\delta)$-differentially private. Scaling $\eps$ and $\delta$
by a factor of $L$,  and substituting for $\beta$, we get the result.
\end{proof}

\subsection{Approximating Hereditary Discrepancy}

In this section we show that our optimal dense case mechanism implies
a polylogarithmic approximation to hereditary discrepancy. In
particular, the mechanism optimal for $\ell_2^2$ error can be used to
approximate $\ell_2$ discrepancy, and the mechanism optimal for
$\ell_\infty$ error can be used to approximate $\ell_\infty$
discrepancy. The $\ell_\infty$ version of hereditary discrepancy is
$\mathsf{NP}$-hard to approximate to within a factor of $3/2$ (proof
in appendix). Showing supeconstant hardness for approximating any
version of hereditary discrepancy, constant hardness for approximating
$\ell_2$ hereditary discrepancy, and determining whether hereditary
discrepancy can be exactly computed in nondeterministic polynomial
time remain open problems. 

Muthukrishnan and Nikolov~\cite{MNstoc} show that hereditary discrepancy gives a lower bound on the error of any mechanism.
\begin{theorem}[\cite{MNstoc}, Corollary 1] There exist constants $\eps,\delta$ such that the following holds: Let $A$ be a $d\times N$ matrix and $\mech$ be an $(\epsilon,\delta)$-differentially private mechanism. Then
\begin{eqnarray*}
\herdisc^{\ell_\infty}(A) \leq O(\log N) \cdot \err^{\ell_\infty}_{\mech}(A)
\end{eqnarray*}
\label{thm:mnherdisc}
\end{theorem}


Fix constants $\eps$ and $\delta$ satisfying Theorem~\ref{thm:mnherdisc}.  Thus the theorem implies that $\herdisc^{\ell_\infty}(A) \leq O(\log N) \cdot \opt^{\ell_\infty}(A)$.
It is easy to see that  for any positive semidefinite diagonal matrix $P$ with $\tr(P^2)=1$, $(\herdisc^{\ell_\infty}(A))^2 \geq \herdisc^{\ell_2^2}(PA)$: indeed for any $S \subseteq [N]$, there exists a coloring $z$ supported on $S$ such that $\|Az\|_{\infty} \leq \herdisc^{\ell_\infty}(A)$. Since $\|PAz\|_2^2 \leq \|Az\|_\infty^2$, the claim follows. Moreover, recall that $\herdisc^{\ell_2^2}(PA) \geq c\cdot\specLB(PA)$, for some absolute constant $c$. Thus
\begin{equation*}
\specLB(PA) \leq (\herdisc^{\ell_\infty}(A))^2
\label{eqn:herdisclb}
\end{equation*}

On the other hand, the mechanism of Theorem~\ref{thm:expected-worstcase-noise} satisfies 
$$(\err^{\ell_\infty}_{\mech_{ew}}(A))^2 \leq O(\log^4 d \log ((\log d)/\delta)) \specLB(PA).$$  
Thus we have 
\stocoption{\begin{align*}
  (\herdisc^{\ell_\infty}&(A))^2 \leq O(\log^2 N) \err^{\ell_\infty}_{\mech_{ew}}(A) \\
  &\leq O(\log^4 d\log^2 N \log ((\log d)/\delta)) \cdot \specLB(PA)
\end{align*}}{
\begin{align*}
  (\herdisc^{\ell_\infty}(A))^2 &\leq O(\log^2 N) \err^{\ell_\infty}_{\mech_{ew}}(A) \\
  &\leq O(\log^4 d\log^2 N \log ((\log d)/\delta)) \cdot \specLB(PA)
\end{align*}}
We have sandwiched $\herdisc^{\ell_\infty}(A)^2$ between two efficiently computable
 quantities that are $O(\log^4 d\log^2 N \log\log d)$ apart.  It follows that
\begin{theorem}
\label{thm:herdiscapprox}
There is a polynomial time algorithm that, given a $d \times N$ real matrix $A$, outputs an $O(\log^{2}d \log N \sqrt{\log \log d})$ approximation to $\herdisc^{\ell_\infty}(A)$.
\end{theorem}

The description of this mechanism $\mech_{ew}$ (which is specified by a distribution over $O(\log d)$ Gaussian noise addition mechanisms) thus serves as an efficiently computable and verifiable witness to an upper bound on the hereditary discrepancy of $A$.

We next give a constrcutive version of this result, by appealing to a result of Bansal~\cite{Bansal10} that shows that the discrepancy of any set system can be constructively upper bounded by $O(\log dN)$ times the hereditary vector discrepancy. The vector discrepancy of a matrix $A$, denoted $\vecdisc(A)$ is defined as the smallest $\lambda$ such that the following semidefinite program is feasible:
\begin{align}
\text{SDP VecDisc: }&A\in \R^{d\times m}&\notag\\
\ai V \ait &\leq \lambda &\forall i \in [d] \notag\\
V_{jj} &\leq 1 & \forall j \in [m]\notag\\
\sum_{j=1}^m V_{jj} &\geq m/8&\label{sdp:vecdisc}\\
V&\succeq 0&\notag
\end{align}
The hereditary vector discrepancy is simply $\hervecdisc(A) \triangleq \max_{S \subseteq [N]} \vecdisc(A|_S)$. We will show that given the mechanism $\mech_{ew}$ of Theorem~\ref{thm:expected-worstcase-noise}, we can construct for any $S$ a solution to the SDP corresponding to $S$. We note that the above SDP is a slight relaxation of the SDP used by Bansal: instead of the constraint $V_{jj} = 1$ for all $j$, we simply require each $V_{jj}$ to be at most one, and that the trace of $V$ is $\Omega(m)$. It is easy to verify that~\cite{Bansal10} implies that:
\begin{theorem}[\cite{Bansal10}]
Suppose that for a matrix $A\in \R^{d\times N}$ and a $\lambda \geq 0$, for every $S \subseteq [N]$ with $rank(A|_S)=|S|$, it is the case that the SDP~\ref{sdp:vecdisc} is feasible for $A|_S$. Them there is polynomial time algorithm that finds a coloring $\chi$ of $[N]$ with discrepancy at most $O(\lambda \log dN)$.
\end{theorem}

We consider a variant of the above SDP, where we drop the $V_{jj} \leq 1$ constraint and require the trace of $V$ to be slightly larger:
\begin{align}
\text{SDP RelVecDisc: }&A\in \R^{d\times m}&\notag\\
\ai V \ait &\leq \lambda &\forall i \in [d] \notag\\
\sum_{j=1}^m V_{jj} &\geq m\label{sdp:relaxedvecdisc}\\
V&\succeq 0&\notag
\end{align}

We first show that it suffices to satisfy the SDP~\ref{sdp:relaxedvecdisc}.
\begin{lemma}
Suppose that for a matrix $A\in \R^{d\times N}$, for every $S \subseteq [N]$ with $rank(A|_S)=|S|$, it is the case that the SDP~\ref{sdp:relaxedvecdisc} is feasible for $A|_S$ and $\lambda$. Then for any $S \subseteq [N]$ with $rank(A|_S)=|S|$, the SDP~\ref{sdp:vecdisc} is feasible for $A|_S$ and $2\lambda$.
\end{lemma}
\begin{proof}
We construct a feasible solution to~\ref{sdp:vecdisc} by repeatedly using solutions to~\ref{sdp:relaxedvecdisc} for different values of the restriction $A|_S$. Fix an $S \subseteq [N]$ with $|S|=m$ and let $S_0=S$. Let $W^0$ be the $d\times m$ zero matrix indexed by $S \subseteq [N]$. Let $V^0$ be a solution to SDP~\ref{sdp:relaxedvecdisc} for $A|_{S_0}$. Let $\gamma_0 = \max_{j} V^0_{jj}$, and let $\bar{V}^{0} = V^{0}/2\gamma_0$.  For each $j,j' \in S_0$, set $W^1_{j j'} = W^{0}_{j j'} + \bar{V}^0_{j j'}$. Finally we set $S_1 = S_0 \setminus \{j: W^1_{jj} \geq \frac 1 4\}$.

 Given $S_i$ such that $|S_i|\geq m/2$, we let $V^i$ be a feasible solution to the SDP for $A|_{S_i}$ and set $\gamma_i = \max_{j} V^i_{jj}$, and $\bar{V}^{i} = V^{i}/2\gamma_i$. For each $j,j' \in S_i$, set $W^{i+1}_{j j'} = W^{i}_{j j'} + \bar{V}^i_{j j'}$. We update $S_{i+1} = S_i \setminus \{j: W^{i+1}_{jj} \geq \frac 1 4\}$. We stop once $|S_i|$ falls below $m/2$, say in iteration $L$. It is easy to see that we delete at least one $j$ from $S_i$ in each step, so that this process converges. 

We claim that $W^L$ is a feasible solution to SDP~\ref{sdp:vecdisc}. Observe that $W^L$ is simply a a non-negative linear combination of $V^i$'s and hence is positive semidefinite. By definition of $\bar{V}$, each diagonal entry is at most $\frac 1 2$, so that the definition of $S_i$ ensures that $W^{L}_{jj} \leq \frac 3 4 < 1$ for all $j$. Moreover, for each $j \in S \setminus S_L$, the entry $W^L_{jj} \geq \frac 1 4$, and there are at least $m/2$ such $j$, so that the trace of $W^L$ is at least $m/8$ as required. Finally, $\ai W^L \ait$ is simply the sum $\sum_{t=0}^{L-1} \frac{\ai V^t \ait}{2\gamma_t}$. Thus it suffices to show that $\sum_{t=1}^{L-1} (2\gamma_t)^{-1} \leq 2$. Note that $\tr(W^{t+1}-W^t) = \tr(\bar{V}^t) \geq (2\gamma_t)^{-1} |S_t| \geq (2\gamma_t)^{-1} m/2$. It follows that $m \geq \tr(W^L) \geq \sum_{t=1}^{L-1} (2\gamma_t)^{-1} m/2$ so that $\sum_{t=1}^{L-1} (2\gamma_t)^{-1} \leq 2$ as needed. The claim follows.
\end{proof}

Finally, we describe how to use the mechanism $\mech_{ew}$ of Theorem~\ref{thm:expected-worstcase-noise} to construct a feasible solution to SDP~\ref{sdp:relaxedvecdisc}.

\begin{lemma}
Suppose that for a matrix $A\in \R^{d\times N}$, we have an oblivious noise \epsdeldp\ mechanism $\mech$ such that $\err^{\ell_\infty}_\mech(A) \leq \kappa$. Then for any $S\subseteq [N]$ with $rank(A|_S)=|S|$, the SDP~\ref{sdp:relaxedvecdisc} is feasible for $A|_S$ with $\lambda = O_{\eps,\delta}(\kappa^2)$.
\end{lemma}
\begin{proof}
Since an \epsdeldp\ mechanism for $A$ also yields an \epsdeldp\ mechanism for any submatrix of $A$, we can assume without loss of generality that $S=[N]$. For a $y\in \R^d$, let $\inv A(y) \triangleq \arg \min_x \|Ax - y\|_{\infty}$. Let $Y$ be a random variable distributed according to $\mech(0)$ and let $X=\inv A(Y)$. Then by the properties of the mechanism $\E[\|AX\|_{\infty}] \leq \E[\|AX-Y\|_{\infty}+\|Y\|_{\infty}] \leq 2\E[\|Y\|_{\infty}] \leq 2\kappa$, since $0$ is a possible value for $\inv A(Y)$, and we picked $X$. Thus if we set $V=E[XX^T]$, then $\ai V \ait \leq 16\kappa^2$ for each $i$.

It remains to lower bound $V_{jj} = \E[(X_j)^2] \geq \Var(X_{j})$. Intuitively, if $\Var(X_{j})$ was small, then $Y$ gives too much information about $X_{j}$, violating differential privacy. Formally, by differential privacy, the statistical distance between $\mech(0)$ and $\mech(e_j)$ is at most $2(\eps+\delta)$. Thus the distribution $Y$ is close to $Y+Ae_j$. But it is easy to check\footnote{We assume that ties in the definition of $\inv A(\cdot)$ are broken at random.} that $\inv A(y+Ae_j) = \inv A(y)+e_j$. Thus the distribution $X$ is close to $X + e_j$. This in turn lower bounds the variance of $X_{j}$ by an absolute constant. Scaling $V$ by a constant factor implies the result.
Finally, we observe that this proof can be easily made constructive. We can take polynomially many samples $x_1,\ldots,x_k$ and setting $V$ to the empirical estimate $\frac{1}{k}\sum_{i=1}^k x_i x_i^T$ instead of $E[XX^T]$. Truncating the distribution so as to reject any $x_i$ with $\|Ax_i\|_{\infty} \geq d^3\kappa$ or with $\|x_i\|_{\infty} \geq d^3$ still preserves the required properties and allows us to use Chernoff bounds to ensure that the empirical estimate satisfies the constraints up to a constant factor.
\end{proof}

\fi

\section{Conclusions}\label{sect:concl}

We have presented near optimal mechanisms for any linear query for dense and sparse databases, under both pure and approximate differential privacy. Our mechanisms are simple and efficient, and it would be instructive to implement them so as to compare them with existing techniques.

Our work uses the hereditary discrepancy lower bound, which holds for
small enough constant $\eps$ and $\delta$. Since our lower bounds
do not get higher as $\delta$ gets smaller, the approximation ratio
has an $O(\sqrt{\log 1/\delta})$ term in it. We leave open the
question of developing better lower bounding techniques, and better
approximation ratios for $(\eps,\delta)$-differentially private
mechanisms. Our work gives $\ell_2^2$ bounds on the error. While we
can translate those bounds to $\ell_\infty$ error bounds for the dense
case, we leave open the question of designing near optimal
$\ell_\infty$ error mechanisms in the sparse case.

\section{Acknowledgements}
The first and second named authors would like to thank Moritz Hardt and
Daniel Dadush for several helpful discussions. The first named author was
partially supported by a Simons Graduate Fellowship, award number
252861. 

\ifstoc
\else
\bibliographystyle{alpha}
\bibliography{privacygeom,noise,privacy}
\fi
\appendix
\stocoption{}{\section{Concentration of Log concave measures}

The following measure concentration inequality is standard. We include a proof below for completeness. We start by stating the Brunn-Minkowski inequality.
\begin{theorem}
Let $\mu$ be a log-concave measure on $\Re^d$, let $\alpha,\beta \geq 0$  be numbers such that $\alpha+\beta=1$, and let $A,B \subseteq \Re^d$ be measurable sets such that the set $\alpha A + \beta B$ is measurable. Then
\[
\mu(\alpha A + \beta B) \geq  (\mu(A))^{\alpha} (\mu(B))^{\beta}
\]
\end{theorem}
This can be used to prove Borell's lemma for arbitrary log concave distributions. We use the proof approach presented  in~\cite{Giannopoulos03}.
 \begin{proofof}{Theorem~\ref{thm:logconcaveborell}}
Observe that 
 \[ \Re^d \setminus A \supseteq (\frac 2 {t + 1}) (\Re^d \setminus tA) + \frac{t - 1}{t + 1} A.\]
 Applying the Brunn-Minkowski inequality and rearranging proves the result.
\end{proofof}}

\section{From Concentration to Expectation}
We repeatedly use the fact that a exponential or better bound on the upper tail implies a bound on the expectation. For completeness, we give a quantitative version with a proof.
\begin{lemma}
\label{lem:conctoexp}
Suppose that for some constants $C,c_1,c_2$, a random variable $X$ satisfies: $\forall \alpha \geq c_1, \Pr[X \geq \alpha C] \leq \exp(- \alpha/c_2)$. Then $\E[X] \leq (c_1+c_2)C$.
\end{lemma}
\begin{proof}
Let $Y=(X-c_1C)/C$. Then
\begin{align*}
  E[Y] &= \int_{0}^{\infty} \Pr[Y \geq \alpha]\,\mathrm{d}\alpha\\
  &\leq \int_{0}^{\infty} \exp(-\alpha/c_2)\,\mathrm{d}\alpha\\
  &= c_2\\
\end{align*}
The claim follows by linearity of expectation.
\end{proof}
\section{Hardness of Approximating Hereditary Discrepancy}

We show constant hardness of approximation
$\herdisc^{\ell_\infty}$. Strong inapproximability results
were previously given for discrepancy, in the $\ell_2$ and
$\ell_\infty$ versions, by Charikar, Newman, and Nikolov~\cite{CNN}. 

Let us define 
$$
\disc^{\ell_\infty} (A) = \min_{x \in \{0, 1\}^n}{\|Ax\|_\infty}
$$
\begin{theorem}\label{thm:herdisc-hardness}
  There exists a family of input matrices $A \in \{0, 1\}^{m \times n}$
for which it is $\mathsf{NP}$-hard to distinguish between the two
cases \textbf{(1)} $\herdisc^{\ell_\infty}(A) \leq 2$ and \textbf{(2)}
$\disc^{\ell_\infty}(A) \geq 3$. 
\end{theorem}

The proof of Theorem~\ref{thm:herdisc-hardness} is  a
straight-forward reduction from the 2-colorability problem for
3-uniform hypergraphs. A maximization version of this problem
(i.e.~maximize the number of bichromatic edges) is also known as
\textsf{Max-E3-Set Splitting} and is equivalent to
\textsf{NotAllEqual-SAT} restricted to inputs with no negated
variables. 

\begin{definition}
  A hypergraph $H = (V, E)$, where $E \subseteq 2^V$, is
  \emph{2-colorable} if and only if there exists a set $T \subseteq V$
  such that for all $e \in E$, $T \cap e \neq e$ and $T \cap e \neq
  \emptyset$. The set $T$ is called a \emph{transversal} of $H$. 
\end{definition}

\begin{lemma}[\cite{schaefer}]
  There exists a family of  3-uniform hypergraphs such that
  deciding whether a hypergraph in the family is 2-colorable is
  $\mathsf{NP}$-complete. 
\end{lemma}

\begin{proofof}{Theorem~\ref{thm:herdisc-hardness}}
  The reduction simply maps a 3-uniform hypergraph to its incidence
  matrix. I.e.~for a hypergraph $H=(V,E)$, where $V = \{v_1, \ldots,
  v_n\}$ and $E = \{e_1, \ldots, e_m\}$, we create a $m \times n$
  matrix $A$, where $A_{ij} = 1$ if $v_j \in e_i$ and $A_{ij} = 0$
  otherwise. Observe that if $H$ is 2-colorable, and this is witnessed
  by a transversal $T$, then $\|(A|_S)x\|_\infty \leq 2$ for all $S
  \subseteq [n]$ and for $x$ defined by $x_i = +1 \Leftrightarrow v_i
  \in T$. On the other hand, if $H$ is not 2-colorable, for any $x \in
  \{+1, -1\}^n$ we have $\|Ax\|_\infty \geq 3$, since otherwise the
  set $T = \{v_i: x_i = +1\}$ would be a transversal.
\end{proofof}

We note that Guruswami proved constant hardness of approximation for
\textsf{Max-E3-Set Splitting}~\cite{venkat-SS}. In particular, he
showed that it is $\mathsf{NP}$-hard to distinguish 2-colorable
3-uniform hypergraphs from hypergraphs for which any coloring with 2
colors leaves at least a $1/20$ fraction of the edges monochromatic.

\end{document}